\documentclass[12pt]{article}

\usepackage{amsmath}
\usepackage{amsthm}
\newtheorem{theorem}{Theorem}
\newtheorem{lemma}{Lemma}

\newtheorem{assumption}{Assumption}
\newtheorem{proposition}{Proposition}
\newtheorem{corollary}{Corollary}
\newtheorem{remark}{Remark}
\usepackage{graphicx}
\usepackage{amssymb, subfigure}
\usepackage{enumerate}
\usepackage{natbib}
\usepackage[table]{xcolor}
\usepackage{diagbox}
\usepackage{url} 
\usepackage[T1]{fontenc}
\usepackage{tikz}

\def\T{{ \mathrm{\scriptscriptstyle T} }}

\newcommand{\blind}{0}

\begin{document}

\def\spacingset#1{\renewcommand{\baselinestretch}%
{#1}\small\normalsize} \spacingset{1}


\if0\blind
{
  \title{\bf Stochastic Generalized Lotka-Volterra Model with An Application to Learning Microbial Community Structures}
  \author{Libai Xu and Ximing Xu\\
    School of statistics and Data Science, Nankai University, Tianjin,\\ 300071, China\\
    libaixuyx@outlook.com  ximing@nankai.edu.cn\\
     Dehan Kong\\
    Department of Statistical Sciences, University of Toronto, Ontario\\ M5S 3G3, Canada\\
    kongdehan@ustat.toronto.edu\\
    Hong Gu and Toby Kenney\hspace{.2cm}\\
	Department of Mathematics and Statistics, Dalhousie University, Halifax,\\ Nova Scotia B3H\\
	hgu@dal.ca tb432381@dal.ca
}
  \maketitle
} \fi

\if1\blind
{
  \bigskip
  \bigskip
  \bigskip
  \begin{center}
   {\bf Stochastic Generalized Lotka-Volterra Model with An Application to Learning Microbial Community Structures}
\end{center}
  \medskip
} \fi

\bigskip
\begin{abstract}
Inferring microbial community structure based on temporal metagenomics
data is an important goal in microbiome studies. The deterministic
generalized Lotka-Volterra differential (GLV) equations have been used
to model the dynamics of microbial data. However, these approaches
fail to take random environmental fluctuations into account, which may
negatively impact the estimates. We propose a new stochastic GLV
(SGLV) differential equation model, where the random perturbations of
Brownian motion in the model can naturally account for the external
environmental effects on the microbial community. We establish new
conditions and show various mathematical properties of the solutions
including general existence and uniqueness, stationary distribution,
and ergodicity. We further develop approximate maximum likelihood
estimators based on discrete observations and systematically
investigate the consistency and asymptotic normality of the proposed
estimators. Our method is demonstrated through simulation studies and
an application to the well-known ``moving picture'' temporal microbial
dataset.
\end{abstract}

\noindent%
\vfill

\newpage
\spacingset{1.5} 
\section{Introduction}
The human microbiome is the collection of all microbes in and on the
human body, and it plays a key role in human health and
disease. Recent advances in high-throughput sequencing technologies
have made it possible to capture microbiome data from human
specimens. There have been many studies investigating the association
between the structure of the microbiome and various outcomes \citep{faust2012microbial,buffie2015precision,cao2017inferring,gibson2016origins,von2018complex}, however, most
studies have focused on the static relationship between the microbiome and
host phenotype, or environmental conditions. The microbiome, while stable in the long term, fluctuates rapidly,
both from external influences, and from its own
dynamics. Understanding the dynamics of the microbiome provides
insight into how the microbiome may change in response to a given
stimulus. This has applications both in predicting adverse effects
caused by the microbiome, and in planning for effective
intervention and remediation.

In \cite{mounier2008microbial}, the deterministic Generalized
Lotka-Volterra (GLV) model was introduced to model inter- and
intraspecies interactions of a cheese microbial
community. \citet{marino2014mathematical} used a deterministic GLV
model to characterize the temporal changes of microbial community in
germfree mice and estimated the interaction effects between the most
abundant operational taxonomic units by least squares estimation. This
method attributed all the effect of environmental variability on the
population dynamics to measurement errors. Figure~\ref{AbundancePlot}
Panel (a) in Appendix~\ref{AppendTable} shows the relative abundance
of Bacteroidaceae from our real data application. The pattern shown is
typical of a system undergoing environmental white noise. However, the
deterministic GLV model fails to account for these environmental
fluctuations \citep{bandyopadhyay2005ratio}, resulting in inaccurate
estimates of the parameters and prediction of future dynamics.
Indeed, \cite{may2019stability} pointed out that due to environmental
fluctuation, the birth rate, carrying capacity, competition
coefficients and other parameters involved with the system exhibit
random fluctuation to different extents. Large amplitude fluctuation
in population may lead to extinction of certain species, which can't
happen in deterministic models.  Hence the stable analysis of the
equilibrium in the deterministic case is not realistic and the
solution may not be reliable. To partially solve this problem,
\citet{stein2013ecological} and \citet{buffie2015precision} added
known external perturbations to incorporate the effect of
environmental variability on the population dynamics. They then
estimated the interaction effects between the most abundant
operational taxonomic units by least squares estimation. In practice,
however, there are many unknown perturbations, that are not accounted
for in their models.

In this paper, we propose a new \textit{Stochastic GLV (SGLV)}
differential equation model to learn interactions between dynamic
microbial communities.  Under the stochastic differential equation,
the abundance of each operational taxonomic unit follows a stochastic
process with the conditional mean following the deterministic GLV
equation, but with an additional noise term following Brownian motion
to represent random external perturbations.  Through simulation
studies, we show that the SGLV model produces better
parameter estimators (in terms of mean squared error) than the
deterministic GLV model estimators from
\citet{marino2014mathematical}. We also show that for real data, the
prediction errors are lower under the SGLV model than under
the deterministic GLV model.

Our paper makes the following contributions. First, we propose a new
SGLV model to study the dynamics of
the microbiome and the interactions between different microbes in the
microbial community. To the best of our knowledge, this is the first
time that the SGLV model is studied
in the statistical literature. Our model is general in the sense that
it can describe competition, mutualism and mixed stochastic systems of
competition and mutualism. Second, we establish a new set of general conditions which guarantee various mathematical properties
  including existence and uniqueness, the stationary distribution, and
  the ergodic property of our SGLV
  model solutions.  The conditions previously derived for competition
  \citep{jiang2012analysis,nguyen2017coexistence,liu2018stationary} and/or mutualism
  \citep{liu2013analysis,liu2015analysis,liu2017stochastic} systems can be regarded as
  special cases of our general conditions. Third, we propose an
approximate maximum likelihood estimator for model parameters, and
study its statistical properties including consistency and asymptotic
normality. There is some literature on approximate maximum likelihood
estimators for stochastic differential equation models. For example,
\cite{dacunha1986estimation}, \cite{yoshida1992estimation}, \cite{genon1993estimation} and \cite{ait2008closed}
consider general one or multi-dimensional diffusion processes and
estimate the unknown parameters using approximate maximum likelihood
estimators.  They have developed consistency and asymptotic normality
under the coercivity and Lipshitz continuity conditions.  However, the
coercivity and Lipshitz continuity conditions do not hold in our
SGLV model, and thus their techniques
can not be applied. Instead, we develop a set of new conditions and
show that our approximate maximum likelihood estimator still satisfies
consistency and asymptotic normality.

The rest of this article is organized as follows. In Section~2, we
introduce our SGLV model, present the mathematical
properties of the solution and propose the approximate
maximum likelihood estimator for the model
parameters. In Section~3, we investigate the statistical properties of
the approximate
maximum likelihood estimator. Simulations are conducted in Section~4 to evaluate the
finite sample performance of the proposed method. In Section~5, we
apply our SGLV model to a temporal microbiome set \citep{caporaso2011moving}. We
conclude with discussion in Section 6. The mathematical proofs are
given in Appendices~\ref{AppendProof} and~\ref{AppendProofTheorem}.

\section{Methodology}
We begin by listing some notation used throughout the paper. Let
$x(t)$ denote the vector of abundances of all $N$ species at time $t$;
$E(x)$ the expectation of random vector $x$; $x^{\T}$ the transpose of
the vector $x$; $\mathbb{R}^{N}_{+}$ the space of $N$-dimensional
positive real numbers and $\mathbb{R}_{\geq 0}$ the non-negative real
numbers. Let $r$ denote the vector of intrinsic growth rates
$(r_{1},...,r_{N})^{\T}$; $\Theta$ the space of drift parameters in
the stochastic differential equation; $\Phi$ the space of diffusion
parameters in the stochastic differential equation; $a\wedge
b=\min\{a,b\}$; $a\vee b=\max\{a,b\}$; $A$ the $N\times N$ matrix with
elements $a_{kl},k,l=1,...,N$; $\sigma^{2}$ the vector of diffusion
parameters ${\rm}(\sigma^{2}_{1},\ldots,\sigma^{2}_{N})^{\T}$; and $C,
C_{1}, C_{2}$ constants which may change between rows. For random
variables, $X_1,X_2,\ldots$, we write $X_n=O_{p}(1)$
(resp. $X_n=o_{p}(1)$) if $X_n$ is bounded (resp. converges) in
probability.  Let $\|\cdot\|$ denote either the Euclidean norm of a
vector or the operator norm of a real matrix.

\subsection{Model}
Consider a collection of $N$ microbes (operational taxonomic units) in
a habitat. Let $x_{k}(t)$ denote the population of microbe $k$ at time
$t$, $ 1\leq k\leq N$ and $ t\geq 0$. To model the complex and dynamic
ecosystem, population dynamics models, especially the GLV model, have
been used for predictive modeling of the intestinal microbiota
\citep{faust2012microbial,stein2013ecological,marino2014mathematical,fisher2014identifying,buffie2015precision}. The
GLV model assumes that the microbe populations follow a set of
ordinary differential equations
\begin{equation}\label{2.1}
 \mathrm{d}x_{k}(t)/\mathrm{d}t = x_{k}(t)\left\{r_{k}+\sum^{N}_{l=1} a_{kl}x_{l}(t)\right\}, \quad k=1,...,N,
\end{equation}
where $a_{kk}=-r_{k}c_{k}<0$, $r_{k}$ denotes the intrinsic growth rate of species $k$,
$c_{k}$ is the coefficient of negative intraspecific interaction representing the inverse of the carrying capacity of the species in isolation and $a_{kl}$ is the interaction coefficient between species. The parameters $r_{k},c_{k}$ and $a_{kl}$ are assumed to be time-invariant.
When $ N=1 $, the model reduces to the classical logistic growth model
\citep{capocelli1975note,roman2012modelling,heydari2014fast,campillo2018parameter}:
$\mathrm{d}x_{1}(t)/\mathrm{d}t = x_{1}(t)\{r_{1}-c_1r_{1}x_{1}(t)\},\nonumber$
where $ c_1=1/K$ and $K$ is the environmental carrying capacity.
Although model~\eqref{2.1} is commonly used in the
literature, many of the systems studied exhibit random fluctuations,
rather than following a deterministic equation
\citep{may2019stability,bandyopadhyay2005ratio}.

To overcome these problems, we propose a \textit{Stochastic GLV
  (SGLV)} differential equation model by perturbing the intrinsic
growth rates of each operational taxonomic unit in
equation~\eqref{2.1}.
\begin{equation}\label{2.4}
\mathrm{d}x_{k}(t)=  x_{k}(t)\left\{r_{k}
+\sum^{N}_{l=1}a_{kl}x_{l}(t)\right\}\mathrm{d}t+\sigma_{k}x_{k}(t)\mathrm{d}B_{k}(t), \quad k=1,...,N,
\end{equation}
where $B(t) = (B_{1}(t),...,B_{N}(t))^{\T}$ is an $N$-dimensional
standard Brownian motion. In particular, $\{B_{k}(t)\}_{1\leq k\leq
  N}$ are mutually independent standard one-dimensional Brownian
motions defined over the complete probability space
$(\Omega,\mathcal{F},\{\mathcal{F}_{t}\}_{t\geq 0},P)$ with filtration
${\mathcal{F}_{t}} $ satisfying the usual conditions.
We can rewrite equation~\eqref{2.4} in vector form as
\begin{equation}\label{1.1}
\mathrm{d}x(t) = {\rm diag}(x(t))\{r+Ax(t)\}\mathrm{d}t + {\rm diag}(\sigma) {\rm diag}(x(t)) \mathrm{d}B(t), t > 0,
\end{equation}
where $r=(r_{1},...,r_{N})^{\T} $,  $A=(a_{kl})_{N\times N}$,
$\sigma=(\sigma_{1},...,\sigma_{N})^{\T}$, $x(t)=(x_1(t), \ldots,
x_N(t))^{\T}$ and $ {\rm diag}(x)$ denotes an
$N\times N$ diagonal matrix with diagonal elements $ x=(x_1, \ldots,
x_N)^{\T} $.
This system is capable of modelling a range of pairwise
interactions. Since $x_{k}(t)$ in SGLV model~\eqref{1.1} represents
the population of the species $k$ at time $t$, we are only interested
in its positive solutions. We therefore focus on the case when
$r_{k}>0$ for all $k$, since cases with $r_{k}<0$ may not have any stable
state in $\mathbb{R}_{+}^{N}$
\citep{may2007theoretical}. Depending on the values
of $a_{kl}~(l\neq k)$, model~\eqref{1.1} can include the following
special cases:
\begin{enumerate}
\item $a_{kl}<0$ for all $l\neq k$: model~\eqref{1.1} is a stochastic competition system.
\item $a_{kl}>0$ for all $l\neq k$: model~\eqref{1.1} is a stochastic mutualism system.
\item $a_{kl}>0$ for some $l\neq k$ and $a_{kl}>0$ for some $l\neq k$ : model~\eqref{1.1} is a
    mixed stochastic system of competition and mutualism.
\end{enumerate}

\subsection{Solution Properties of Stochastic Generalized Lotka-Volterra Model}
In this subsection,
we study the theoretical properties for the
solution to equation~\eqref{1.1}. The main results are presented in
Propositions~\ref{pro2.1}, \ref{pro2.2}, \ref{pro2.3} and
Corollary~\ref{pro2.4}, with the proofs deferred to Appendix~\ref{AppendProof}.
We need the following conditions:
\begin{assumption}
\label{assumptionA1}
The initial value $x(0)=(x_{1}(0),\ldots,x_{N}(0))^{\T}\in \mathbb{R}^{N}_{+}$, $r_{k}-\sigma^{2}_{k}/2>0,\sigma_{k}>0$,  and $A=(a_{kl})_{N\times N}$ is  non-positive definite.
\end{assumption}
\begin{assumption}
 \label{assumptionA2} For some $\phi\geq 4$, the elements of $A$ satisfy
 $$a_{kk}+\phi(\phi+1)^{-1}\sum_{l=1}^N(a_{kl}\lor
0)+(\phi+1)^{-1}\sum_{l=1}^N(a_{lk}\lor 0)<0.$$
\end{assumption}
\begin{assumption}
\label{assumptionA3}
Each element of the vector $\tilde{x}=(\tilde{x}_{1},...,\tilde{x}_{N})^{\T}=-A^{-1}\left(r-\sigma^{2}/2\right)$ is positive, i.e.,  $\tilde{x}_{k}>0$ for $k=1,...,N$.
\end{assumption}
\begin{assumption}
\label{assumptionA4}
There exist positive constants $c_1$,$\ldots$, $c_N$ such that
for $k=1,...,N$,
$$\sum^{N}_{i=1}c_{i}\sigma^{2}_{i}\tilde{x}_{i} < -\left[2c_{k}a_{kk}
   +\sum_{l\neq k}\{c_{k}|a_{kl}|
   +c_{l}|a_{lk}|\}   \right]\tilde{x}^{2}_{k}.$$
\end{assumption}

\begin{proposition}\label{pro2.1}
For $t\geq0$, under Assumption~\ref{assumptionA1}, there is a unique solution $x(t)$ to SGLV model~\eqref{1.1} and  $x(t)\in \mathbb{R}^{N}_{+}$ almost surely.
\end{proposition}

Most previous literature on the existence of a unique global solution
to a stochastic differential equation requires conditions including the linear growth condition
and local Lipschitz condition
\citep{arnold1974stochastic,friedman2010stochastic,oksendal2013stochastic,liu2015stochastic}.
However, the coefficients
of~model \eqref{1.1} do not satisfy the linear growth condition. Therefore,
we develop new techniques to show that the solution of~model \eqref{1.1}
can't explode in finite time. Particularly we establish a new set
of conditions which guarantee existence, uniqueness and positivity of
the global solution to stochastic differential equation~\eqref{1.1}.
We also prove some bounds on the moments of the solution, which will
be used to prove Theorems \ref{FixedT} and \ref{NoFixT}.
\begin{proposition}\label{pro2.2}
Under Assumptions~\ref{assumptionA1} and \ref{assumptionA2}, there exist positive constants $C_{1}$ and $C_{2}$ such that for any initial value $x_{0}\in \mathbb{R}_{+}^{N}$,
the solution of SGLV~model \eqref{1.1} has the properties
\begin{equation}\label{limsup}
\limsup_{t\rightarrow+\infty}\sum^{N}_{k=1}E\{x^{\theta}_{k}(t)\}\leq C_{1},\quad
  \limsup_{t\rightarrow+\infty}\sum^{N}_{k=1}t^{-1}E\left\{\int^{t}_{0}x^{\theta}_{k}(s)ds\right\}\leq C_{2},\nonumber
\end{equation}
for any $0\leq\theta\leq 4$.
\end{proposition}
Finally we need to show that the solution has a
stationary distribution and is ergodic.
\begin{proposition}\label{pro2.3}
Under Assumptions~\ref{assumptionA1}, \ref{assumptionA3} and \ref{assumptionA4}, there is a stationary distribution for the solution of SGLV~model~\eqref{1.1}, and it has the ergodic property.
\end{proposition}

\begin{corollary}{}\label{pro2.4}
Under Assumptions~\ref{assumptionA1}, \ref{assumptionA3} and \ref{assumptionA4}, for any Borel measurable function $f(\cdot):\mathbb{R}^{N}_{+}\rightarrow \mathbb{R}$,
which is integrable
with respect to the density, $\pi(\cdot)$, of the stationary distribution, the solution of SGLV~model \eqref{1.1} has the property,
\begin{equation}\label{2.11}
\lim_{t\rightarrow+\infty}t^{-1}\int^{t}_{0}f(x(s))ds=\int_{\mathbb{R}^{N}_{+}}f(x)\pi(dx).\nonumber
\end{equation}
\end{corollary}
This proposition and corollary are key to obtaining the limit of the
average of the continuous log-likelihood function with respect to time, and also crucial to the proof of our asymptotic theory.

\subsection{Parameter Estimation for the Stochastic Generalized Lotka-Volterra Model}\label{parameterestimation}
In this section, we develop approximate maximum likelihood estimators of the parameters $ \{r_k, 1\leq k\leq N\} $, $ \{a_{kl}, 1\leq k\leq N, 1\leq l\leq N\} $ and $ \{\sigma_k, 1\leq k\leq N\} $.

Let $u_{k}(t)=\log x_{k}(t)$ for $ t\in [0, T] $ and $R_{k}=r_{k}-\sigma^{2}_{k}/2$. By It\^{o}'s formula we have
\begin{equation}\label{logSDE}
\mathrm{d}u_{k}(t)= \left[R_{k}+\sum^{N}_{l=1}a_{kl}\exp\{u_{l}(t)\}\right]\mathrm{d}t+\sigma_{k}\mathrm{d}B_{k}(t),
\end{equation}
from which the true log-likelihood function can be derived for
continuously-observed data. However, in practice the data $\{x_k(t_i), 1\leq i\leq n, 1\leq k\leq N\}$ are only observed at a
sequence of discrete time points $ 0=t_1<t_2<\ldots<t_n=T$. We approximate stochastic differential equation \eqref{logSDE} over intervals $[t_{i},t_{i+1}]$
using Euler's approximation. Let $\{\epsilon_{k,i},k=1, \cdots,N, i=1,
\cdots, n-1 \}$ denote independent and identically distributed standard normal distributions, $\Delta_{i,t}=t_{i+1}-t_i,\Delta_{i}u_{k}=u_{k}(t_{i+1})-u_{k}(t_{i})$,
\begin{equation}\label{disSDE}
\Delta_{i}u_{k} \approx\left[R_{k}+\sum^{N}_{l=1}a_{kl}\exp\{u_{l}(t_i)\}\right]\Delta_{i,t}
+\sigma_{k}\Delta^{1/2}_{i,t}\epsilon_{k,i}.
\end{equation}
Let $\mathcal{F}_{i}$ denote a sigma algebra generated by $\{u_{1}(t_i),...,u_{N}(t_i)\}$.
By the Markov property,
the approximate likelihood of the $(i+1)$th observation is
\begin{eqnarray}
&&f(u(t_{i+1})|\mathcal{F}_{i})\nonumber\\
&=& \prod^{N}_{k=1} [(2\pi)^{-1/2} \Delta^{-1/2}_{i,t}\sigma_{k}^{-1}\exp\{-
  2^{-1}\sigma^{-2}_{k}\Delta^{-1}_{i,t}(\Delta_{i}u_{k}- [R_{k}+\sum^{N}_{l=1}a_{kl}\exp\{u_{l}(t_i)\}]\Delta_{i,t})^{2}\}].\nonumber
\end{eqnarray}
Ignoring constant terms, the approximate log-likelihood function is
\begin{eqnarray}\label{EulerEstimate}
&& \mathcal{L}_{n,T}(\vartheta,\sigma^{2})\nonumber\\
&=& -\sum^{N}_{k=1}\left\{(n-1)\log\sigma^{2}_{k}
+\sum^{n-1}_{i=1}\sigma^{-2}_{k}\Delta^{-1}_{i,t}\left(\Delta_{i}u_{k}-\left[R_{k}+\sum^{N}_{l=1}a_{kl}
\exp\{u_{l}(t_{i})\}\right]
\Delta_{i,t}\right)^{2}\right\}.\nonumber
\end{eqnarray}
Let $\vartheta=(\vartheta_{1},...,\vartheta_{N})^{\T}\in
\mathbb{R}^{N(N+1)} $ with
$\vartheta_{k}=(r_{k},a_{k1},...,a_{kN})^{\T}\in \mathbb{R}^{N+1}$ be
the drift parameters, and
$\sigma^{2}=(\sigma^{2}_{1},...,\sigma^{2}_{N})^{\T}$ the diffusion parameters. The approximate log-likelihood is
$$  \mathcal{L }_{n,T}(\vartheta,\sigma^{2})=\log(L(u(t_{1}),...,u(t_{n})|\mathcal{F}_{0})) = \mathcal{G}_{n,T}(\sigma^{2})+\ell_{n,T}(\vartheta),$$
where
$\mathcal{G}_{n,T}(\sigma^{2}) =-(n-1)\sum^{N}_{k=1}\log \sigma_{k}-\sum^{N}_{k=1}\sum_{i=1}^{n-1}2^{-1}\sigma^{-2}_{k}\Delta^{-1}_{i,t}\{\Delta_{i}u_{k}\}^{2}\nonumber
$
does not depend on the drift parameters $\vartheta$, and $\ell_{n,T}(\vartheta)=\sum^{N}_{k=1}\sigma^{-2}_{k}\ell_{n,T}(\vartheta_{k})$
is the discrete version of the continuous log-likelihood, with $\ell_{n,T}(\vartheta_{k})$ defined as
$$  \ell_{n,T}(\vartheta_{k}) =  \sum_{i=1}^{n-1}\Delta_{i}u_{k}\left[R_{k}+\sum^{N}_{l=1}a_{kl}\exp\{u_{l}(t_i)\}\right]-
\sum_{i=1}^{n-1}\left[R_{k}+\sum^{N}_{l=1}a_{kl}\exp\{u_{l}(t_i)\}\right]^{2}\Delta_{i,t}/2.$$
The approximate maximum likelihood estimators of $\vartheta$ and $\sigma^{2}$, denoted by $\widehat{\vartheta}_{n,T}$ and $\widehat{\sigma}^{2}_{n,T}$ respectively, can be solved from
$\partial \mathcal{L}_{n,T}(\vartheta,\sigma^{2})/\partial \vartheta =0$~ and~$
\partial \mathcal{L}_{n,T}(\vartheta,\sigma^{2})/\partial \sigma^{2} =0.\nonumber$
As $\mathcal{G}_{n,T}(\sigma^{2})$ does not depend on $\vartheta$,
$\widehat{\vartheta}_{n,T}=(\hat{\vartheta}_{1,n,T},...,\hat{\vartheta}_{N,n,T})^{\T}$
can be solved from $\partial \ell_{n,T}(\vartheta)/\partial
\vartheta=0$ directly. The closed forms of the approximate maximum
likelihood estimators of SGLV~model
\eqref{1.1}, $\widehat{\sigma}^{2}_{n,T}$ and
$\widehat{\vartheta}_{n,T}$ are
\begin{align}\label{MLEestimators}
  \widehat{\sigma}^{2}_{k,n,T} &= (n-1)^{-1}\sum^{n-1}_{i=1}\left(\Delta_{i}u_{k}-\left[\hat{R}_{k}+\sum^{N}_{l=1}\hat{a}_{kl}\exp\{u_{l}(t_{i})\}\right]\Delta_{i,t}\right)^{2}\Delta^{-1}_{i,t},\nonumber\\
\hat{r}_{k,n,T} &=\hat{R}_{k}+\widehat{\sigma}^{2}_{k,n,T}/2,\quad \hat{R}_{k} =T^{-1}\left[u_{k}(t_n)-u_{k}(t_1)-\sum^{n-1}_{i=1}\sum^{N}_{l=1}\Delta_{i,t}\hat{a}_{kl} \exp\{u_{l}(t_{i})\}\right],\nonumber\\
  \hat{a}_{kp} &=\left\{\begin{array}{ll}
  \left(L^{-1}M\right)_{kp},&\textrm{if }k=p,\\
  -\left(L^{-1}M\right)_{kp},&\textrm{otherwise,}\end{array}\right.
\end{align}
where
\begin{align*}
L_{l,s} &= \sum_{i=1}^{n-1}\Delta_{i,t} \left[\sum_{i=1}^{n-1}\exp\{u_{l}(t_i)+u_{s}(t_i)\}\Delta_{i,t}\right]
-\left[\sum_{i=1}^{n-1}\exp\{u_{l}(t_i)\}\Delta_{i,t}\right]\left[ \sum_{i=1}^{n-1}\exp\{u_{s}(t_i)\}\Delta_{i,t}\right],\\
  M_{k,p} &= \{u_{k}(t_{n})-u_{k}(t_1)\}\sum_{i=1}^{n-1}\exp\{u_{p}(t_i)\}\Delta_{i,t}-
  \sum^{n-1}_{i=1}\Delta_{i,t}\left[\sum_{i=1}^{n-1}\{u_{k}(t_{i+1})-u_{k}(t_i)\}\exp\{u_{p}(t_i)\}\right].
\end{align*}

The estimators of interaction coefficients $\hat{a}_{kl}$ in
formula~\eqref{MLEestimators} may be positive or negative, representing a
stimulatory or antagonistic microbial interaction.

\section{Asymptotic Theory}\label{theory}
In this section, we derive consistency and asymptotic normality
properties for the approximate maximum likelihood estimators in formula~\eqref{MLEestimators}. Our case
is more challenging than conventional theories for maximum likelihood estimators, because we
need to account for the approximation error in addition to the usual
statistical error.  To prove the theorems, we first show that the
differences between the approximate maximum likelihood estimators
$\widehat{\sigma}^{2}_{k,n,T}$ and $\widehat{\vartheta}_{n,T}$
obtained in formula~\eqref{MLEestimators} and the continuous-time maximum likelihood estimators,
denoted by $\widehat{\sigma}^{2}_{k,T}$ and $\widehat{\vartheta}_{T}$,
are neglegible.  Our next step is to show that $
\widehat{\sigma}^{2}_{k,n,T}$ and $\widehat{\vartheta}_{n,T}$ converge
to the true values, denoted by $ \sigma^{2}$ and $ \vartheta^{0}$. The
detailed proof is deferred to Appendix~\ref{AppendProofTheorem}.

Before we present our main theorem, we first define the maximum likelihood estimators $
\widehat{\sigma}^{2}_{k,T}$ and $\widehat{\vartheta}_{T}$ based on the
true continuous likelihood function. The estimators $
\widehat{\sigma}^{2}_{k,T}$ are uniquely determined through the
following equations
$$\label{trueMLEsigma}
  \widehat{\sigma}^{2}_{k,T}=\lim_{n\rightarrow+\infty}T^{-1
  }\sum^{2^{n}}_{i=1}\{u_{k}(iT/2^{n})-u_{k}((i-1)T/2^{n})\}^{2}, k=1,...,N.\nonumber
$$
The estimation of drift parameters $\hat{\vartheta}_T$ is performed by
maximum likelihood.

Let $\Delta_{\max}=\max_{i=0,\ldots,n-1}\Delta_{i,t}$. We need one of
the following conditions in addition to
Assumptions~\ref{assumptionA1}--\ref{assumptionA4}:
\begin{assumption}
\label{assumptionA5}
$T$ is fixed, and $\Delta_{\max}\rightarrow0 $.
\end{assumption}
\begin{assumption}
\label{assumptionA6}
(I) $T\rightarrow+\infty$ and $\Delta_{\max}\rightarrow0 $; (II) $T\rightarrow+\infty$ and
$T\Delta_{\max}\rightarrow 0$.
\end{assumption}

\begin{theorem}\label{FixedT}
Under Assumptions~\ref{assumptionA1}, \ref{assumptionA2}, and
\ref{assumptionA5}, we have that, conditional on the maximum likelihood
estimators lying in some compact parameter space $\mathcal{K}$,

(i) $\|\hat{\vartheta}_{n,T}-\hat{\vartheta}_{T}\|=O_{p}\left(\Delta^{1/2}_{\max}\right)$.

(ii) For any $k\in\{1,...,N\}$,
$(n/2)^{1/2}\sigma^{-2}_{k}(\widehat{\sigma}^{2}_{k,n,T}-\sigma^{2}_{k}) \rightarrow N(0,1).\nonumber$
\end{theorem}
\begin{remark}\label{remarktheorem1}
Theorem \ref{FixedT} gives the convergence rate and asymptotic
normality of diffusion parameter estimates for fixed observation time
$T$. For the drift parameters, this theorem only gives the rates of
$\|\hat{\vartheta}_{n,T}-\hat{\vartheta}_{T}\|$, the discrepancy
between the approximate maximum likelihood estimators and the
continuous-time maximum likelihood estimators, and this discrepancy decreases
at rate $ \Delta^{1/2}_{\max} $. However, this theorem does not provide
any guarantee about $\|\hat{\vartheta}_{n,T}-\vartheta_{0}\|$ because
the continuous-time maximum likelihood estimators $
\hat{\vartheta}_{T} $ may not converge to the true parameter values $
\vartheta_{0} $ due to the fact that the ergodicity of $x(t)$ does not
apply to finite $T$. Therefore, for fixed $ T $, our approximate
maximum likelihood estimators might be biased regardless of the value
of $n$. The bias comes from two parts: the gap
$\|\hat{\vartheta}_{T}-\vartheta_{0}\|$ and Euler's
approximation. These parts decrease when $T$ increases and
$\Delta_{\max}$ decreases respectively.
\end{remark}

To make the bias asymptotically negligible, we need to consider infinite observation time $ T $. In particular, the following
theorem adapts Theorem~\ref{FixedT} to ergodic diffusions. The
ergodicity is essential for the asymptotic theory to hold as
$T\rightarrow+\infty$.
We define the matrix
$\mathcal{I}={\rm diag}\left(\sigma^{-2}_{1}I_{1},...,\sigma^{-2}_{N}I_{N}\right)$, where
\begin{equation}\label{Information}
  I_{k}(\vartheta_{k}) =-\partial^{2} \ell_{\vartheta^{0}_{k}}(\vartheta_{k})/(\partial \vartheta_k)^{2}=\int_{\mathbb{R}^{N}_{+}}\mu_{k}(u)P_{\vartheta^{0}}(du)
\end{equation}
where the continuous log-likelihood of the $k$th variable is
\begin{align*}
  \ell_{\vartheta^{0}_{k}}(\vartheta_{k})
=&\int_{\mathbb{R}^{N}}\left[R_{k}+\sum^{N}_{l=1}a_{kl}\exp\{u_{l}(t)\}\right]\left[\tilde{R}^{0}_{k}
  +\sum^{N}_{l=1}a^{0}_{kl}\exp\{u_{l}(t)\}\right]P_{\vartheta_{0}}(du_{k})\\
&\qquad-\frac{1}{2}\int_{\mathbb{R}^{N}}\left[R_{k}+\sum^{N}_{l=1}a_{kl}\exp\{u_{l}(t)\}\right]^{2}P_{\vartheta_{0}}(du_{k}),
\end{align*}
$P_{\vartheta^{0}}(du)$ is the stationary distribution of $u(t)$ and
$$\mu_{k}={\rm diag}\left\{\left(\frac{\partial}{\partial \vartheta_{k}}\left[R_{k}+\sum^{N}_{l=1}a_{kl}\exp\{u_{l}(t)\}\right]\right)^{\T}\left(\frac{\partial}{\partial \vartheta_{k}}\left[R_{k}+\sum^{N}_{l=1}a_{kl}\exp\{u_{l}(t)\}\right]\right)\right\}_{\tilde{N}\times\tilde{N}}.$$

\begin{theorem}\label{NoFixT}
Under Assumptions~\ref{assumptionA1}--\ref{assumptionA4} and~\ref{assumptionA6}(I), we have
\begin{align*}
  \|\hat{\vartheta}_{n,T}-\hat{\vartheta}_{T}\|=O_{p}(\Delta^{1/2}_{\max}),\quad
  \|\hat{\vartheta}_{n,T}-\vartheta^{0}\|=o_{p}(1).
\end{align*}
If we further assume Assumption~\ref{assumptionA6}(II), then
\begin{align*}
  T^{1/2}(\hat{\vartheta}_{n,T}-\vartheta^{0}) \rightarrow N(0, \mathcal{I}^{-1}(\vartheta^{0})),\quad
(n/2)^{1/2}\sigma^{-2}_{k}(\widehat{\sigma}^{2}_{k,n,T}-\sigma^{2}_{k}) \rightarrow N(0,1)
\end{align*}
where $\mathcal{I}(\vartheta)={\rm diag}\{\sigma^{-2}_{1} I_{1}(\vartheta_{1}),...,\sigma^{-2}_{N}I_{N}(\vartheta_{N})\}$, $I_{k}(\vartheta_{k})$ is defined in (\ref{Information}), $k=1,...,N$.
\end{theorem}
\begin{remark}
The theorem states that approximate maximum likelihood estimators are consistent and asymptotically
efficient when the observation time goes to infinity and maximum step size $\Delta^{1/2}_{\max}$ goes to zero.
\end{remark}

\section{Simulation Study}
In this section, we investigate the finite sample performance of the
proposed method. We set $N=5$, with initial values
$x_{0}=(0.5,0.15,0.13,0.05,0.04)^{\T}$. We simulate the temporal
dynamics using Euler's approximation~\eqref{disSDE} with a time step
$\Delta=0.01$.  To simulate a discrete sample, we simulate $n-1$
sample time steps $\Delta_{i,t}=t_{n,i+1}-t_{n,i}$, $1\leq i\leq n-1$
independently from the distribution $\Delta_{i,t}=\{0.1,0.3,0.5\}$
with probability $\{0.7,0.2,0.1\}$ respectively.
The $n$ time points $t_{n,1}<t_{n,2}<\ldots<t_{n,n}$ are then
generated by $t_{n,j} =\sum_{i=1}^{j-1} \Delta_{i,t}$ (so
$t_{n,1}=0$). We consider the following two parameter settings
including both positive and negative interaction coefficients:
\begin{description}
  \item
  Case 1: $r=(1,1.5,2,1.5,2)^{\T},\sigma=(0.1,0.1,0.1,0.1,0.1)^{\T},{\scalebox{0.9} \bf{A}=\left(
  \begin{array}{rrrrr}
    -2 & -2.5 & -2  & 1 & 1 \\
    1 & -6  & -2 & 3 & -1 \\
    -1 & -2 & -5  & 1 & -1 \\
    -1 & 0.5 & 0.1  & -10 & 1\\
    -1.5 & -2  & -2 & 2 & -9 \\
  \end{array}
\right).}$
  \item Case 2:~~Same setting as Case~1 except that $\sigma=(1,1,1,1,1)^{\T}$.
\end{description}

The interaction coefficients and growth rates are designed to be
similar to the coefficients estimated from the real data in
Section~\ref{RealData}. For each case, we consider three sample sizes,
$n\in\{300,500,1000\}$. We compare the
estimators based on the SGLV model~\eqref{1.1} with the
least squares estimation~\citep{cao2017inferring,bucci2016mdsine,gibson2016origins} of
deterministic GLV model~\eqref{2.1}. We report
the mean squared error
of the parameter estimates. As the deterministic GLV  Model does not estimate the
diffusion parameter $\sigma$, we do not report it for this approach.

We report the mean squared errors of the parameters for
Case~1, $n=1000$ based on 1000 simulations in Table~\ref{Case1-1000},
\renewcommand\arraystretch{1}
\renewcommand\arraystretch{0.6}
\begin{table}[!htbp]
\caption{Simulation results for Case 1, $n=1000$,
	$\Delta_{i,t}\in\{0.1,0.3,0.5\}$ based on 1000 Monte Carlo samples:
	mean squared errors (standard error) of estimates of $\{a_{kl},
	1\leq k\leq 5, 1\leq l\leq 5\}$ and $\{r_k, 1\leq k\leq 5\}$. We
	write `--' to indicate values not estimated by the method; `GLV' is the
	deterministic GLV model method; and `SGLV' is
	our method.}\label{Case1-1000}\begin{center}
		\resizebox{\textwidth}{28mm}{
			\setlength{\tabcolsep}{0.5mm}{
				\begin{tabular}{r|rrrrr|rr}
					\hline
					\hline
					\multicolumn{1}{l|}{Method}&\multicolumn{5}{c}{Mean squared error (standard error)}&\multicolumn{1}{r}{}& \multicolumn{1}{r}{}\\
					\hline
					&  \multicolumn{1}{c}{$a_{11}$} & \multicolumn{1}{c}{$a_{12}$} & \multicolumn{1}{c}{$a_{13}$} & \multicolumn{1}{c}{$a_{14}$} & \multicolumn{1}{c|}{$a_{15}$} & \multicolumn{1}{c}{$r_{1}$} &\multicolumn{1}{c}{$\sigma^{2}_{1}$ } \\
					\multicolumn{1}{c|}{GLV}&  $ 0.179(0.008)$ & $0.642(0.025)$ & $0.534(0.021)$ & $ 1.093(0.046)$ &
					$0.901(0.042)$ & $0.127(0.006)$ &\multicolumn{1}{c}{$-$} \\
					\multicolumn{1}{c|}{SGLV}& $ 0.136(0.006)$ & $0.438(0.018) $ & $0.358(0.015)$ & $ 0.841(0.038)$ &
					$0.705(0.032)$ & $0.089(0.004)$&$ 0.232(3.287)\times 10^{-6}$\\
					\cline{1-8}
					&  \multicolumn{1}{c}{$a_{21}$} & \multicolumn{1}{c}{$a_{22}$} & \multicolumn{1}{c}{$a_{23}$} & \multicolumn{1}{c}{$a_{24}$} & \multicolumn{1}{c|}{$a_{25}$} &\multicolumn{1}{c}{ $r_{2}$} &\multicolumn{1}{c}{$\sigma^{2}_{2}$} \\
					\multicolumn{1}{c|}{GLV}&  $ 0.155(0.007)$ & $1.788(0.039)$ & $0.801(0.022)$ & $2.087(0.064)$ &
					$0.836(0.037)$ & $0.204(0.007)$ &\multicolumn{1}{c}{$-$}\\
					\multicolumn{1}{c|}{SGLV}&  $ 0.131(0.006)$ & $1.048(0.027)$ & $0.522(0.016)$ & $ 1.428(0.049)$
					& $0.685(0.031)$ & $0.129(0.005)$ &$3.362(2.121)\times 10^{-6}$ \\
					\cline{1-8}
					&  \multicolumn{1}{c}{$a_{31}$} & \multicolumn{1}{c}{$a_{32}$} & \multicolumn{1}{c}{$a_{33}$} & \multicolumn{1}{c}{$a_{34}$} & \multicolumn{1}{c|}{$a_{35}$} & \multicolumn{1}{c}{$r_{3}$} &\multicolumn{1}{c}{$\sigma^{2}_{3}$} \\
					\multicolumn{1}{c|}{GLV}&  $ 0.213(0.008)$ & $0.877(0.026)$ & $1.534(0.031)$ & $0.873(0.035)$
					& $0.737(0.031)$ & $0.328(0.008)$ &\multicolumn{1}{c}{$-$} \\
					\multicolumn{1}{c|}{SGLV}&  $ 0.157(0.007)$ & $0.575(0.019)$ & $0.856(0.021)$ & $0.692(0.029)$
					& $0.620(0.026)$ & $0.198(0.006)$ &$4.162(2.396)\times 10^{-6}$ \\
					\cline{1-8}
					&   \multicolumn{1}{c}{$a_{41}$} & \multicolumn{1}{c}{$a_{42}$} & \multicolumn{1}{c}{$a_{43}$}&  \multicolumn{1}{c}{$a_{44}$} & \multicolumn{1}{c|}{$a_{45}$} & \multicolumn{1}{c}{$r_{4}$} &\multicolumn{1}{c}{$\sigma^{2}_{4}$} \\
					\multicolumn{1}{c|}{GLV}&  $ 0.252(0.009)$ & $0.295(0.013)$ & $0.176(0.008)$ & $ 5.604(0.103) $
					& $0.682(0.029)$ & $0.126(0.005)$ &\multicolumn{1}{c}{$-$} \\
					\multicolumn{1}{c|}{SGLV}&  $ 0.175(0.007)$ & $0.240(0.010)$ & $0.162(0.007)$ & $3.177(0.073)$
					& $0.606(0.026)$ & $0.084(0.003)$ &$4.962(2.954)\times 10^{-6}$ \\
					\cline{1-8}
					&  \multicolumn{1}{c}{$a_{51}$} & \multicolumn{1}{c}{$a_{52}$} & \multicolumn{1}{c}{$a_{53}$} &  \multicolumn{1}{c}{$a_{54}$} & \multicolumn{1}{c|}{$a_{55}$} & \multicolumn{1}{c}{$r_{5}$} & \multicolumn{1}{c}{$\sigma^{2}_{5}$}\\
					\multicolumn{1}{c|}{GLV}&  $ 0.291(0.011)$ & $0.873(0.027)$ & $0.864(0.024)$ & $1.547(0.053)$
					& $2.886(0.077)$ & $0.308(0.009)$ &\multicolumn{1}{c}{$-$} \\
					\multicolumn{1}{c|}{SGLV}&  $ 0.200(0.008)$ & $0.589(0.020)$ & $0.572(0.017)$ & $1.113(0.041)$
					& $1.715(0.052)$ & $0.190(0.006)$ &$2.707(2.033)\times 10^{-6}$ \\
					\cline{1-8}
					\hline
		\end{tabular}}}
	\end{center}
\end{table}

From the results, we see that our method outperforms the deterministic
GLV approach
in terms of mean squared errors
for a large majority of the parameter
estimates.

Simulation results for Case 1 and Case 2, $n\in\{300,500,1000\}$ are presented
in Appendix~\ref{AppendTable}, Tables~\ref{Case1-delta0-300}--\ref{Case2-delta0-1000}. From
the mean squared error
of parameters in
Tables~\ref{Case1-delta0-300}--\ref{Case2-delta0-1000}, we can see
that our method always outperforms the deterministic GLV approach.
The increase in diffusion parameter in Case~2 leads to an increase in MSE of
parameter estimates due to larger fluctuations in the data.
We also notice that MSE of parameter estimates decreases with
sample size increasing when  diffusion parameters are small
which is consistent with our theory.

\section{Real Data Application}\label{RealData}

We apply our method to the moving picture dataset
\citep{caporaso2011moving}. The dataset consists of samples from 4 body
sites on 2 individuals. We focus on the samples from the faeces of
person~2. We choose the faeces site because the gut is more sheltered
from external influences than other body sites, and we therefore
expect it to provide better insights into the internal dynamics of the
system. We choose person~2 because there are more
observations. \citet{cai2017learning} found that there is a shift in person~2's
gut microbiome during the study. Since we are interested in the
equilibrium dynamics of the microbiome, we analyse only the data prior
to this shift.

The data are collected on a daily basis. The time intervals between
consecutive data points can vary, ranging from 1 day to 9 days.
Microbiome data are subject to large multiplicative noise, commonly
referred to as {\em sequencing depth} which is generally thought to be
unrelated to the microbial dynamics. We therefore analyse the
proportions of each operational taxonomic unit in the community,
rather than the total counts. There are other approaches to deal with
the sequencing depth issue, see
\citep{weiss2017normalization} for a summary, but in
the absence of a clear consensus over which method is best for our
application, we have chosen a simple and widely-used method for our
analysis. We do not expect the choice of normalisation to have a large
impact on our results.

To obtain more stability and interpretability, we group the
operational taxonomic units at the family level. There are 107
family-level operational taxonomic units observed in the data. We
select the five most abundant families from the data, namely
Bacteroidaceae, Ruminococcaceae, Lachnospiraceae, Porphyromonadaceae,
and an unspecified family from the class Bacteroidales. In
Appendix~\ref{AppendTable}, Figure~\ref{AbundancePlot}(a) shows the
temporal relative abundance of the family Bacteroidaceae, while
Figure~\ref{AbundancePlot}(b) shows a log transformation of the
relative abundance of the family Bacteroidaceae and one-step
predictions under our estimated SGLV
model.  \renewcommand\arraystretch{1}
\begin{table}[htbp]
	\begin{center}
	\caption{Estimated interaction coefficients, growth rates and
		diffusion coefficients for the 5 most abundant families in
		person~2$'$s gut. The significantly non-zero interaction
		coefficients (at the 5\% level), excluding the diagonal entries and
		growth rates, which are all significantly non-zero, are highlighted
		in blue.}\label{GutFamily} \setlength{\tabcolsep}{1.5mm}{
			\begin{tabular}{r|rrrrr|rr}
				\hline
				\hline
				\diagbox{~}{$a_{kp}$}{~~~~~~~~~~~} & \multicolumn{1}{c}{B.} & \multicolumn{1}{c}{R.} & \multicolumn{1}{c}{L.} & \multicolumn{1}{c}{P.} & \multicolumn{1}{c}{B$'$.} & \multicolumn{1}{c}{$r_i$} & \multicolumn{1}{c}{$\sigma^{2}_i$} \\
				\hline
				\multicolumn{1}{l|}{Bacteroidaceae}       & $ -1.997$ & \cellcolor{blue!30!white!70}{$-2.212$}& \cellcolor{blue!30!white!70}{$-2.035$} &
				$2.713$ & $0.176$  & 1.573&  0.464 \\
				\multicolumn{1}{l|}{Ruminococcaceae}      & $0.008$ & $ -3.957$ & \cellcolor{blue!30!white!70}{$-1.954$} & $ \cellcolor{blue!30!white!70}{3.352}$ &$-0.602$  & 0.768 &  0.167  \\
				\multicolumn{1}{l|}{Lachnospiraceae}      & $-0.364$ & \cellcolor{blue!30!white!70}{$-1.633$} & $ -5.481$ & $ 1.059$ & $-0.604$ & 1.290  &  0.376\\
				\multicolumn{1}{l|}{Porphyromonadaceae}   & $-0.308$ & $0.036$ & $0.017$ & $-10.371$ & $0.410$ & 0.761 &  0.262   \\
				\multicolumn{1}{l|}{Bacteroidales (unsp.)}& \cellcolor{blue!30!white!70}{$-1.322$} &
				$-1.584$ & \cellcolor{blue!30!white!70}{$-3.322$} & $2.015$ & $-8.984$  & 1.795&  0.564 \\
				\hline
				\hline
			\end{tabular}
		}
	\end{center}
\end{table}

The estimates of the interaction coefficients, growth rates and
diffusion coefficients are given in Table~\ref{GutFamily}.
Note that the estimated values of $r_{i}$ and $\sigma^{2}_{i} (i=1,...,5)$ satisfy
Assumption~\ref{assumptionA1}, and the main theoretical results including Proposition~\ref{pro2.2}, Theorem~\ref{FixedT} and Theorem~\ref{NoFixT} hold for the estimated interaction coefficients in Table~\ref{GutFamily}.
The estimated interaction coefficients indicate a mixture of
competition and mutualism, with Porphyromonadaceae stimulating growth
of the other families. This is consistent with 
\cite{darveau2012porphyromonas},
who found that, in the oral microbiome, one species from the family
Porphyromonadaceae can manipulate the host immune system, allowing
colonization by microbes that would usually be suppressed by the
immune system. A similar effect in the gut microbiome could explain
the patterns estimated here. The competition between other families
also seems biologically plausible, partially because of the
compositionality of the data, meaning the growth of one operational
taxonomic unit results in a reduced proportion of other operational
taxonomic units.

We construct confidence intervals by applying
Theorem~\ref{NoFixT} to construct the confidence interval
$\left[\hat{a}_{kl}-1.96\hat{\sigma}_{k}\sqrt{\{{\rm
    diag}(\hat{I}^{-1}_{k})\}_{l}},
  \hat{a}_{kl}+1.96\hat{\sigma}_{k}\sqrt{\{{\rm
    diag}(\hat{I}^{-1}_{k})\}_{l}}\right]$. In Table~\ref{GutFamily}, we
highlight the significantly non-zero interaction coefficients using
these confidence intervals. These significant interactions are plotted
in Figure~\ref{GutFamilyNetwork}.

The 5\%
confidence intervals
are reported in Appendix~\ref{AppendTable},
Table~\ref{GutFamilyInterval5}.  For comparison, the estimated
parameters under a least-squares GLV model
method are given in Table~\ref{GutFamilyTableGLV}. Corresponding 5\%
confidence intervals are in Table~\ref{GutFamilyGLVConfidence} of Appendix~\ref{AppendTable}. Since
Asymptotic confidence intervals are not available for the
deterministic GLV model, we use bootstrap
confidence intervals in this case. There is some consistency in the
estimated parameters under the two models, with many estimates being
positive or negative in both models. The significant interactions all
have the same sign under both models. The estimated interaction
coefficients and estimated growth rates are generally larger in
magnitude under our SGLV model.

\begin{figure}[htbp]
\centering
\caption{Significant interactions between abundant families. Blue
  arrows represent negative interactions. Thickness of arrow is
  proportional to interaction strength. Radius of circle is
  proportional to growth rate.}\label{GutFamilyNetwork}

\begin{center}
  \includegraphics[width=3.0in]{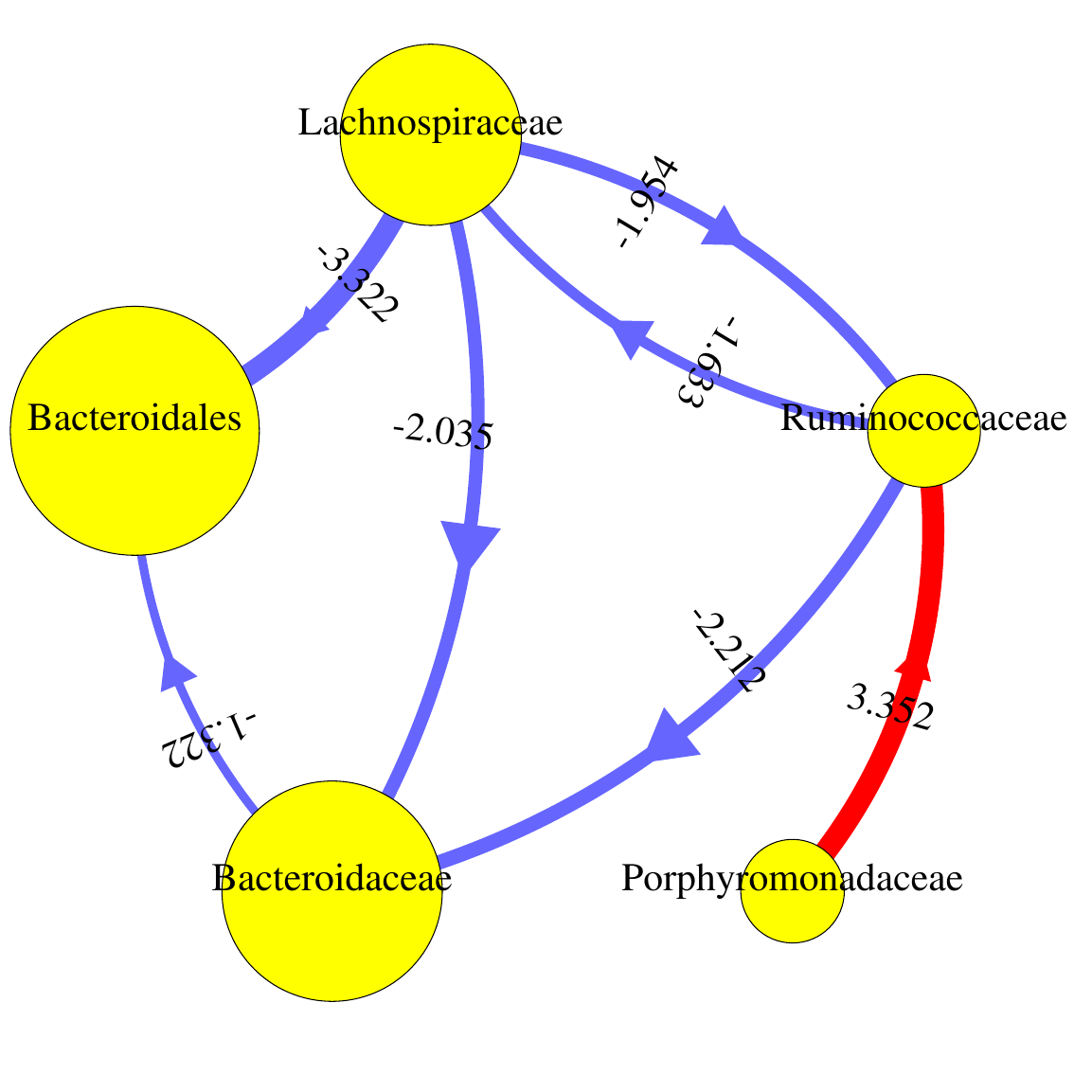}
\end{center}

\end{figure}

We also compare the predictive performance of our approach and the
least squares approach of the deterministic GLV
model. We divide the data randomly into training and test data
sets. We use a proportion $(1-1/k)$ of the data as the training set,
and the remaining $1/k$ portion of the data as the test data. We
consider $ k=24, 12$, and $8$. For each $k$, we average the results
over 100 random splits. We report the one-step prediction errors
obtained from both approaches in Table~\ref{MPE1}. The mean squared
prediction error is defined as ${\textrm{mean~squared~prediction~error}}:=
m^{-1}\sum^{m}_{i=1}(\hat{u}_{i}-u_{i})^{2},$ where $u_{i}$ is the
observed value at $t_{i}$, $\hat{u}_{i}$ is computed from the training
model based on the previous log-observed data point $u_{i-1}$ at time
$t_{i-1}$. From the results, our method outperforms the deterministic
GLV model approach in prediction. Although the
purpose of the analysis is interpretation rather than prediction,
improved predictive accuracy suggests a better fitting model, so the
conclusions drawn are more likely to be accurate.
\begin{table}[h]
\caption{Cross-validated mean squared prediction error (standard error) of
  one-step predictions on moving picture gut data, using a proportion
  $1/k$ of the data as test data, averaged over 100 training-test
  splits. `GLV' is the deterministic GLV
  model method, and `SGLV' is our method.  }\label{MPE1}
\begin{center}
\setlength{\tabcolsep}{2mm}{
\begin{tabular}{c|ccc}
\hline
    &\multicolumn{3}{c}{mean~squared~prediction~error~(standard~error)} \\
\hline
     $k$ & $24$ & $12$  &$8$\\
      $\rm GLV$ & $1.346(0.105)$& $1.495(0.089)$ &$1.762(0.096)$  \\
      $\rm SGLV$ &$1.206(0.093)$& $1.410(0.083) $ & $1.610(0.087)$\\
\hline
\end{tabular}
}
\end{center}

\end{table}

\newpage

\section{Conclusion}
\label{sec:conc}
Motivated by analysis of microbiome data, we proposed a new SGLV model to study
temporal dynamics of the microbiome, and derived a new set of general conditions to
guarantee various mathematical properties including existence and uniqueness of a solution,
the bounds of moments, stationarity and ergodicity. Further, we have developed
an approximate maximum likelihood estimator for our model and established a theoretical
guarantee for the consistency and asymptotic normality of our AMLE. We have
demonstrated the efficiency of our methods using simulations and real microbiome data
applications.

There are a number of important directions for future work. Firstly,
we have assumed that the number of OTUs, $N$, is fixed. In many
applications, $N$ can be much larger than the sample size, $n$, in which
case, regularization methods are needed to ensure good
estimation. Incorporating regularisation into our estimators is highly
non-trivial because the penalized estimates of $A$ may not be
non-positive definite. We may add additional constraints to guarantee
non-positive definiteness of $A$. However, this may lead to challenges
in both computation and theory, and therefore we leave this topic for
future research.  For the asymptotic consistency of our method, it is
essential to assume that the time difference between consecutive
samples converges to 0. When this assumption is violated, i.e. time
difference between consecutive samples does not converge to 0, the
finite sample properties of our AMLE are unclear, and it is an
interesting problem to explore in the future.

\bibliographystyle{Chicago}
\bibliography{Bibliography-MM-MC}

\begin{thebibliography}{}

\bibitem[\protect\citeauthoryear{Ait-Sahalia et~al.}{Ait-Sahalia
  et~al.}{2008}]{ait2008closed}
Ait-Sahalia, Y. et~al. (2008).
\newblock Closed-form likelihood expansions for multivariate diffusions.
\newblock {\em The Annals of Statistics\/}~{\em 36\/}(2), 906--937.

\bibitem[\protect\citeauthoryear{Arnold}{Arnold}{1974}]{arnold1974stochastic}
Arnold, L. (1974).
\newblock Stochastic differential equations.
\newblock {\em New York\/}.

\bibitem[\protect\citeauthoryear{Bandyopadhyay and Chattopadhyay}{Bandyopadhyay
  and Chattopadhyay}{2005}]{bandyopadhyay2005ratio}
Bandyopadhyay, M. and J.~Chattopadhyay (2005).
\newblock Ratio-dependent predator--prey model: effect of environmental
  fluctuation and stability.
\newblock {\em Nonlinearity\/}~{\em 18\/}(2), 913.

\bibitem[\protect\citeauthoryear{Bucci, Tzen, Li, Simmons, Tanoue, Bogart,
  Deng, Yeliseyev, Delaney, Liu, et~al.}{Bucci et~al.}{2016}]{bucci2016mdsine}
Bucci, V., B.~Tzen, N.~Li, M.~Simmons, T.~Tanoue, E.~Bogart, L.~Deng,
  V.~Yeliseyev, M.~L. Delaney, Q.~Liu, et~al. (2016).
\newblock Mdsine: Microbial dynamical systems inference engine for microbiome
  time-series analyses.
\newblock {\em Genome biology\/}~{\em 17\/}(1), 121.

\bibitem[\protect\citeauthoryear{Buffie, Bucci, Stein, McKenney, Ling,
  Gobourne, No, Liu, Kinnebrew, Viale, et~al.}{Buffie
  et~al.}{2015}]{buffie2015precision}
Buffie, C.~G., V.~Bucci, R.~R. Stein, P.~T. McKenney, L.~Ling, A.~Gobourne,
  D.~No, H.~Liu, M.~Kinnebrew, A.~Viale, et~al. (2015).
\newblock Precision microbiome reconstitution restores bile acid mediated
  resistance to clostridium difficile.
\newblock {\em Nature\/}~{\em 517\/}(7533), 205--208.

\bibitem[\protect\citeauthoryear{Cai, Gu, and Kenney}{Cai
  et~al.}{2017}]{cai2017learning}
Cai, Y., H.~Gu, and T.~Kenney (2017).
\newblock Learning microbial community structures with supervised and
  unsupervised non-negative matrix factorization.
\newblock {\em Microbiome\/}~{\em 5\/}(1), 110.

\bibitem[\protect\citeauthoryear{Campillo, Joannides, and
  Larramendy-Valverde}{Campillo et~al.}{2018}]{campillo2018parameter}
Campillo, F., M.~Joannides, and I.~Larramendy-Valverde (2018).
\newblock Parameter identification for a stochastic logistic growth model with
  extinction.
\newblock {\em Communications in Statistics-Simulation and Computation\/}~{\em
  47\/}(3), 721--737.

\bibitem[\protect\citeauthoryear{Cao, Gibson, Bashan, and Liu}{Cao
  et~al.}{2017}]{cao2017inferring}
Cao, H.-T., T.~E. Gibson, A.~Bashan, and Y.-Y. Liu (2017).
\newblock Inferring human microbial dynamics from temporal metagenomics data:
  Pitfalls and lessons.
\newblock {\em BioEssays\/}~{\em 39\/}(2), 1600188.

\bibitem[\protect\citeauthoryear{Capocelli and Ricciardi}{Capocelli and
  Ricciardi}{1975}]{capocelli1975note}
Capocelli, R. and L.~Ricciardi (1975).
\newblock A note on growth processes in random environment.
\newblock {\em Biological cybernetics\/}~{\em 18\/}(2), 105--109.

\bibitem[\protect\citeauthoryear{Caporaso, Lauber, Costello, Berg-Lyons,
  Gonzalez, Stombaugh, Knights, Gajer, Ravel, Fierer, et~al.}{Caporaso
  et~al.}{2011}]{caporaso2011moving}
Caporaso, J.~G., C.~L. Lauber, E.~K. Costello, D.~Berg-Lyons, A.~Gonzalez,
  J.~Stombaugh, D.~Knights, P.~Gajer, J.~Ravel, N.~Fierer, et~al. (2011).
\newblock Moving pictures of the human microbiome.
\newblock {\em Genome biology\/}~{\em 12\/}(5), R50.

\bibitem[\protect\citeauthoryear{Dacunha-Castelle and
  Florens-Zmirou}{Dacunha-Castelle and
  Florens-Zmirou}{1986}]{dacunha1986estimation}
Dacunha-Castelle, D. and D.~Florens-Zmirou (1986).
\newblock Estimation of the coefficients of a diffusion from discrete
  observations.
\newblock {\em Stochastics: An International Journal of Probability and
  Stochastic Processes\/}~{\em 19\/}(4), 263--284.

\bibitem[\protect\citeauthoryear{Darveau, Hajishengallis, and Curtis}{Darveau
  et~al.}{2012}]{darveau2012porphyromonas}
Darveau, R., G.~Hajishengallis, and M.~Curtis (2012).
\newblock Porphyromonas gingivalis as a potential community activist for
  disease.
\newblock {\em Journal of dental research\/}~{\em 91\/}(9), 816--820.

\bibitem[\protect\citeauthoryear{Faust and Raes}{Faust and
  Raes}{2012}]{faust2012microbial}
Faust, K. and J.~Raes (2012).
\newblock Microbial interactions: from networks to models.
\newblock {\em Nature Reviews Microbiology\/}~{\em 10\/}(8), 538--550.

\bibitem[\protect\citeauthoryear{Fisher and Mehta}{Fisher and
  Mehta}{2014}]{fisher2014identifying}
Fisher, C.~K. and P.~Mehta (2014).
\newblock Identifying keystone species in the human gut microbiome from
  metagenomic timeseries using sparse linear regression.
\newblock {\em PloS one\/}~{\em 9\/}(7).

\bibitem[\protect\citeauthoryear{Friedman}{Friedman}{2010}]{friedman2010stochastic}
Friedman, A. (2010).
\newblock Stochastic differential equations and applications.
\newblock In {\em Stochastic Differential Equations}, pp.\  75--148. Springer.

\bibitem[\protect\citeauthoryear{Genon-Catalot and Jacod}{Genon-Catalot and
  Jacod}{1993}]{genon1993estimation}
Genon-Catalot, V. and J.~Jacod (1993).
\newblock On the estimation of the diffusion coefficient for multi-dimensional
  diffusion processes.
\newblock In {\em Annales de l'IHP Probabilit{\'e}s et statistiques},
  Volume~29, pp.\  119--151.

\bibitem[\protect\citeauthoryear{Gibson, Bashan, Cao, Weiss, and Liu}{Gibson
  et~al.}{2016}]{gibson2016origins}
Gibson, T.~E., A.~Bashan, H.-T. Cao, S.~T. Weiss, and Y.-Y. Liu (2016).
\newblock On the origins and control of community types in the human
  microbiome.
\newblock {\em PLoS computational biology\/}~{\em 12\/}(2).

\bibitem[\protect\citeauthoryear{Heydari, Lawless, Lydall, and
  Wilkinson}{Heydari et~al.}{2014}]{heydari2014fast}
Heydari, J., C.~Lawless, D.~A. Lydall, and D.~J. Wilkinson (2014).
\newblock Fast bayesian parameter estimation for stochastic logistic growth
  models.
\newblock {\em Biosystems\/}~{\em 122}, 55--72.

\bibitem[\protect\citeauthoryear{Jiang, Ji, Li, and O'Regan}{Jiang
  et~al.}{2012}]{jiang2012analysis}
Jiang, D., C.~Ji, X.~Li, and D.~O'Regan (2012).
\newblock Analysis of autonomous lotka--volterra competition systems with
  random perturbation.
\newblock {\em Journal of Mathematical Analysis and Applications\/}~{\em
  390\/}(2), 582--595.

\bibitem[\protect\citeauthoryear{Liu and Wang}{Liu and
  Wang}{2013}]{liu2013analysis}
Liu, M. and K.~Wang (2013).
\newblock Analysis of a stochastic autonomous mutualism model.
\newblock {\em Journal of Mathematical Analysis and Applications\/}~{\em
  402\/}(1), 392--403.

\bibitem[\protect\citeauthoryear{Liu and Zhu}{Liu and
  Zhu}{2018}]{liu2018stationary}
Liu, M. and Y.~Zhu (2018).
\newblock Stationary distribution and ergodicity of a stochastic hybrid
  competition model with l{\'e}vy jumps.
\newblock {\em Nonlinear Analysis: Hybrid Systems\/}~{\em 30}, 225--239.

\bibitem[\protect\citeauthoryear{Liu, Chen, and Hu}{Liu
  et~al.}{2015}]{liu2015analysis}
Liu, Q., Q.~Chen, and Y.~Hu (2015).
\newblock Analysis of a stochastic mutualism model.
\newblock {\em Communications in Nonlinear Science and Numerical
  Simulation\/}~{\em 29\/}(1-3), 188--197.

\bibitem[\protect\citeauthoryear{Liu, Jiang, Shi, Hayat, and Alsaedi}{Liu
  et~al.}{2017}]{liu2017stochastic}
Liu, Q., D.~Jiang, N.~Shi, T.~Hayat, and A.~Alsaedi (2017).
\newblock Stochastic mutualism model with l{\'e}vy jumps.
\newblock {\em Communications in Nonlinear Science and Numerical
  Simulation\/}~{\em 43}, 78--90.

\bibitem[\protect\citeauthoryear{Liu and R{\"o}ckner}{Liu and
  R{\"o}ckner}{2015}]{liu2015stochastic}
Liu, W. and M.~R{\"o}ckner (2015).
\newblock {\em Stochastic partial differential equations: an introduction}.
\newblock Springer.

\bibitem[\protect\citeauthoryear{Marino, Baxter, Huffnagle, Petrosino, and
  Schloss}{Marino et~al.}{2014}]{marino2014mathematical}
Marino, S., N.~T. Baxter, G.~B. Huffnagle, J.~F. Petrosino, and P.~D. Schloss
  (2014).
\newblock Mathematical modeling of primary succession of murine intestinal
  microbiota.
\newblock {\em Proceedings of the National Academy of Sciences\/}~{\em
  111\/}(1), 439--444.

\bibitem[\protect\citeauthoryear{May, McLean, et~al.}{May
  et~al.}{2007}]{may2007theoretical}
May, R., A.~R. McLean, et~al. (2007).
\newblock {\em Theoretical ecology: principles and applications}.
\newblock Oxford University Press on Demand.

\bibitem[\protect\citeauthoryear{May}{May}{2019}]{may2019stability}
May, R.~M. (2019).
\newblock {\em Stability and complexity in model ecosystems}, Volume~1.
\newblock Princeton university press.

\bibitem[\protect\citeauthoryear{Mounier, Monnet, Vallaeys, Arditi, Sarthou,
  H{\'e}lias, and Irlinger}{Mounier et~al.}{2008}]{mounier2008microbial}
Mounier, J., C.~Monnet, T.~Vallaeys, R.~Arditi, A.-S. Sarthou, A.~H{\'e}lias,
  and F.~Irlinger (2008).
\newblock Microbial interactions within a cheese microbial community.
\newblock {\em Appl. Environ. Microbiol.\/}~{\em 74\/}(1), 172--181.

\bibitem[\protect\citeauthoryear{Nguyen and Yin}{Nguyen and
  Yin}{2017}]{nguyen2017coexistence}
Nguyen, D.~H. and G.~Yin (2017).
\newblock Coexistence and exclusion of stochastic competitive lotka--volterra
  models.
\newblock {\em Journal of differential equations\/}~{\em 262\/}(3), 1192--1225.

\bibitem[\protect\citeauthoryear{Oksendal}{Oksendal}{2013}]{oksendal2013stochastic}
Oksendal, B. (2013).
\newblock {\em Stochastic differential equations: an introduction with
  applications}.
\newblock Springer Science \& Business Media.

\bibitem[\protect\citeauthoryear{Rom{\'a}n-Rom{\'a}n and
  Torres-Ruiz}{Rom{\'a}n-Rom{\'a}n and Torres-Ruiz}{2012}]{roman2012modelling}
Rom{\'a}n-Rom{\'a}n, P. and F.~Torres-Ruiz (2012).
\newblock Modelling logistic growth by a new diffusion process: Application to
  biological systems.
\newblock {\em Biosystems\/}~{\em 110\/}(1), 9--21.

\bibitem[\protect\citeauthoryear{Stein, Bucci, Toussaint, Buffie, R{\"a}tsch,
  Pamer, Sander, and Xavier}{Stein et~al.}{2013}]{stein2013ecological}
Stein, R.~R., V.~Bucci, N.~C. Toussaint, C.~G. Buffie, G.~R{\"a}tsch, E.~G.
  Pamer, C.~Sander, and J.~B. Xavier (2013).
\newblock Ecological modeling from time-series inference: insight into dynamics
  and stability of intestinal microbiota.
\newblock {\em PLoS computational biology\/}~{\em 9\/}(12).

\bibitem[\protect\citeauthoryear{von Bronk, G{\"o}tz, and Opitz}{von Bronk
  et~al.}{2018}]{von2018complex}
von Bronk, B., A.~G{\"o}tz, and M.~Opitz (2018).
\newblock Complex microbial systems across different levels of description.
\newblock {\em Physical biology\/}~{\em 15\/}(5), 051002.

\bibitem[\protect\citeauthoryear{Weiss, Xu, Peddada, Amir, Bittinger, Gonzalez,
  Lozupone, Zaneveld, V{\'a}zquez-Baeza, and Birmingham}{Weiss
  et~al.}{2017}]{weiss2017normalization}
Weiss, S., Z.~Z. Xu, S.~Peddada, A.~Amir, K.~Bittinger, A.~Gonzalez,
  C.~Lozupone, J.~R. Zaneveld, Y.~V{\'a}zquez-Baeza, and A.~Birmingham (2017).
\newblock Normalization and microbial differential abundance strategies depend
  upon data characteristics.
\newblock {\em Microbiome\/}~{\em 5\/}(1), 27.

\bibitem[\protect\citeauthoryear{Yoshida}{Yoshida}{1992}]{yoshida1992estimation}
Yoshida, N. (1992).
\newblock Estimation for diffusion processes from discrete observation.
\newblock {\em Journal of Multivariate Analysis\/}~{\em 41\/}(2), 220--242.

\end{thebibliography}


\begin{thebibliography}{7}
	\expandafter\ifx\csname natexlab\endcsname\relax\def\natexlab#1{#1}\fi
	\bibitem[{Brown \& Hewitt(1975)}]{BH:1975}
	\textsc{Brown, B. M., \& Hewitt, J. I.}(1975).
	\newblock {Asymptotic likelihood theory for diffusion processes}.
	\textsl{J. Appl. Probab.}, \textbf{12(2)},228-238.
	
	\bibitem[{Huzak(2018)}]{MH:2017}
	\textsc{Huzak, M.}(2018).
	\newblock{Estimating a class of diffusions from discrete observations via approximate maximum likelihood method}.
	\textsl{Statistics,} \textbf{52(2)},239-272.
	
	\bibitem[{Ka\'{s}minskii(2011)}]{RZH:1980}
	\textsc{Ka\'{s}minskii, R. }2011.
	\newblock{Stochastic Stability of Differential Equations}(Vol.66).
	\textsl{Springer Science \& Business Media}.
	
	\bibitem[Zhu \& Yin(2007)]{ZY:2007}
	\textsc{Zhu, C., \& Yin, G.}(2007).
	\newblock{Asymptotic properties of hybrid diffusion systems}.
	\textsl{SIAM J. Control. Optim.} \textbf{46(4)},1155-1179.
	
\end{thebibliography}

\clearpage
\bigskip

\appendix
\renewcommand{\theequation}{S\arabic{equation}}
\renewcommand{\thefigure}{S\arabic{figure}}
{\tiny }\renewcommand{\theexample}{S\arabic{example}}
\renewcommand{\thetable}{S\arabic{table}}
\setcounter{assumption}{0}
\setcounter{proposition}{0}
\setcounter{corollary}{0}

\section{Additional Simulation Results}  \label{AppendTable}
In this section, we present additional simulation results for various
scenarios. In particular, we include simulation results with various
sample sizes $n\in\{300,500,1000\}$, for Cases~1 and 2 from Section~4
in main paper. Time step sizes are still set at $\Delta_{i,t}\in
\{0.1,0.3,0.5\}$ with probabilities $0.7$, $0.2$ and $0.1$
respectively.  All results are based on 1,000 simulations, with mean
square error results compared for the approximate maximum likelihood
estimators of our SGLV model, and the least squares estimate under
deterministic GLV Differential Equation.
\renewcommand\arraystretch{0.6}
\begin{table}[!htbp]
	\caption{Simulation results for Case 1, $n=300$, based on 1000 simulations: mean squared errors (standard error) of estimates of
		$\{a_{kl}, 1\leq k\leq 5, 1\leq l\leq 5\}$ and $\{r_k, 1\leq k\leq
		5\}$. }\label{Case1-delta0-300}\begin{center}
		\resizebox{\textwidth}{28mm}{ \setlength{\tabcolsep}{0.5mm}{
				\begin{tabular}{r|rrrrr|rr}
					\hline
					\hline
					\multicolumn{1}{c|}{Method}&\multicolumn{5}{c}{Mean squared error (standard error)}&\multicolumn{1}{r}{}& \multicolumn{1}{r}{}\\
					\hline
					&  \multicolumn{1}{c}{$a_{11}$} & \multicolumn{1}{c}{$a_{12}$} & \multicolumn{1}{c}{$a_{13}$} & \multicolumn{1}{c}{$a_{14}$} & \multicolumn{1}{c|}{$a_{15}$} & \multicolumn{1}{c}{$r_{1}$} &\multicolumn{1}{c}{$\sigma^{2}_{1}$} \\
					\multicolumn{1}{c|}{GLV}&  $ 0.584(0.027)$ & $1.271(0.057) $ & $1.040(0.048)$ & $
					3.305(0.144)$ & $3.362(0.162)$&$0.332(0.016)$&\multicolumn{1}{c}{$-$} \\
					\multicolumn{1}{c|}{SGLV} &  $ 0.442(0.021)$ & $0.878(0.037) $ & $0.720(0.031)$ & $
					2.456(0.114)$ & $2.329(0.105)$ & $0.242(0.011)$& $ 0.765(3.590) \times 10^{-6}$\\
					\cline{1-8}
					&  \multicolumn{1}{c}{$a_{21}$} & \multicolumn{1}{c}{$a_{22}$} & \multicolumn{1}{c}{$a_{23}$} & \multicolumn{1}{c}{$a_{24}$} & \multicolumn{1}{c|}{$a_{25}$} &\multicolumn{1}{c}{$r_{2}$} &\multicolumn{1}{c}{$\sigma^{2}_{2}$} \\
					\multicolumn{1}{c|}{GLV}&  $ 0.483(0.021)$ & $1.685(0.057) $ & $1.291(0.049)$ & $
					3.392(0.135)$ & $2.500(0.115)$&$0.343(0.015)$ &\multicolumn{1}{c}{$-$} \\
					\multicolumn{1}{c|}{SGLV}&  $ 0.430(0.019)$ & $1.093(0.042) $ & $0.952(0.038)$ & $
					2.566(0.110)$ & $2.124(0.098)$&$0.272(0.012)$ &$4.189(3.163)\times 10^{-6}$\\
					\cline{1-8}
					&  \multicolumn{1}{c}{$a_{31}$} & \multicolumn{1}{c}{$a_{32}$} & \multicolumn{1}{c}{$a_{33}$} & \multicolumn{1}{c}{$a_{34}$} & \multicolumn{1}{c|}{$a_{35}$} & \multicolumn{1}{c}{$r_{3}$} &\multicolumn{1}{c}{$\sigma^{2}_{3}$} \\
					\multicolumn{1}{c|}{GLV}&  $ 0.485(0.020)$ & $1.272(0.051)$ & $1.412(0.047)$ & $
					2.050(0.089)$ & $
					2.350(0.106)$ & $0.426(0.016)$ &\multicolumn{1}{c}{$-$} \\
					\multicolumn{1}{c|}{SGLV}&  $ 0.424(0.018)$ & $0.989(0.041)$ & $0.871(0.033)$ & $
					1.937(0.086)$ & $2.182(0.099)$ & $0.305(0.012)$ &$4.757(3.327) \times 10^{-6}$\\
					\cline{1-8}
					&   \multicolumn{1}{c}{$a_{41}$} & \multicolumn{1}{c}{$a_{42}$} &
					\multicolumn{1}{c}{$a_{43}$}&  \multicolumn{1}{c}{$a_{44}$} &
					\multicolumn{1}{c|}{$a_{45}$} & \multicolumn{1}{c}{$r_{4}$} &
					\multicolumn{1}{c}{$\sigma^{2}_{4}$} \\
					\multicolumn{1}{c|}{GLV}&  $ 0.493(0.019)$ & $0.817(0.039)$ & $0.579(0.027)$ &
					$ 5.111(0.167) $ & $2.131(0.099)$ & $0.245(0.010)$ &\multicolumn{1}{c}{$-$} \\
					\multicolumn{1}{c|}{SGLV}&  $ 0.418(0.018)$ & $0.718(0.033)$ & $0.548(0.027)$ &
					$3.029(0.114)$ & $1.947(0.089)$ & $0.212(0.010)$ &$5.602(3.896)\times 10^{-6}$\\
					\cline{1-8}
					&  \multicolumn{1}{c}{$a_{51}$} & \multicolumn{1}{c}{$a_{52}$} &
					\multicolumn{1}{c}{$a_{53}$} &  \multicolumn{1}{c}{$a_{54}$} &
					\multicolumn{1}{c|}{$a_{55}$} & \multicolumn{1}{c}{$r_{5}$} &
					\multicolumn{1}{c}{$\sigma^{2}_{5}$}\\
					\multicolumn{1}{c|}{GLV}&  $ 0.555(0.023)$ & $1.274(0.053)$ & $1.285(0.048)$ &
					$3.289(0.134)$ & $3.299(0.120)$ & $0.408(0.015)$ &\multicolumn{1}{c}{$-$} \\
					\multicolumn{1}{c|}{SGLV}&   $ 0.454(0.021)$ & $0.942(0.040)$ & $0.930(0.036)$ &
					$ 2.523(0.106)$ & $2.221(0.095)$ & $0.291(0.012)$ &$3.361(2.945) \times 10^{-6}$\\
					\cline{1-8}
					\hline
		\end{tabular}}}
	\end{center}
\end{table}
\renewcommand\arraystretch{0.6}
\begin{table}[!htbp]
	\caption{Simulation results for Case~1, $n=500 $ based on 1000 simulations: mean squared errors (standard error) of estimates of $\{a_{kl}, 1\leq k\leq 5, 1\leq l\leq 5\}$ and $\{r_k, 1\leq k\leq 5\}$.}\label{Case1-500}\begin{center}
		\resizebox{\textwidth}{28mm}{
			\setlength{\tabcolsep}{0.5mm}{
				\begin{tabular}{r|rrrrr|rr}
					\hline
					\hline
					\multicolumn{1}{c|}{Method}&\multicolumn{5}{c}{Mean squared error(standard error)}&\multicolumn{1}{r}{}& \multicolumn{1}{r}{}\\
					\hline
					&  \multicolumn{1}{c}{$a_{11}$} & \multicolumn{1}{c}{$a_{12}$} & \multicolumn{1}{c}{$a_{13}$} & \multicolumn{1}{c}{$a_{14}$} & \multicolumn{1}{c|}{$a_{15}$} & \multicolumn{1}{c}{$r_{1}$} &\multicolumn{1}{c}{$\sigma^{2}_{1}$ } \\
					\multicolumn{1}{c|}{GLV}&  $ 0.361(0.016)$ & $0.978(0.043)$ & $0.733(0.030)$ &
					$ 2.214(0.100)$ & $1.805(0.082)$ & $0.221(0.010)$ &\multicolumn{1}{c}{$-$} \\
					\multicolumn{1}{c|}{SGLV}& $ 0.279(0.013)$ & $0.672(0.028) $ & $0.498(0.021)$ & $
					1.715(0.073)$ & $1.358(0.063)$ & $0.162(0.007)$&$ 0.467(3.402)\times 10^{-6}$\\
					\cline{1-8}
					&  \multicolumn{1}{c}{$a_{21}$} & \multicolumn{1}{c}{$a_{22}$} & \multicolumn{1}{c}{$a_{23}$} & \multicolumn{1}{c}{$a_{24}$} & \multicolumn{1}{c|}{$a_{25}$} &\multicolumn{1}{c}{$r_{2}$} &\multicolumn{1}{c}{$\sigma^{2}_{2}$} \\
					\multicolumn{1}{c|}{GLV}&  $0.279(0.012)$ & $1.781(0.050) $ & $1.027(0.035)$ & $
					2.682(0.100)$ & $1.521(0.066)$ & $0.270(0.011)$ &\multicolumn{1}{c}{$-$} \\
					\multicolumn{1}{c|}{SGLV}&  $0.243(0.011)$ & $1.005(0.032) $ & $0.688(0.025)$ & $1.954(0.077)$ & $1.288(0.058)$ & $0.184(0.007)$ &$3.640(2.586)\times 10^{-6}$ \\
					\cline{1-8}
					&  \multicolumn{1}{c}{$a_{31}$} & \multicolumn{1}{c}{$a_{32}$} & \multicolumn{1}{c}{$a_{33}$} &
					\multicolumn{1}{c}{$a_{34}$} & \multicolumn{1}{c|}{$a_{35}$} & \multicolumn{1}{c}{$r_{3}$} &\multicolumn{1}{c}{$\sigma^{2}_{3}$} \\
					\multicolumn{1}{c|}{GLV}&  $0.308(0.014)$ & $1.076(0.037) $ & $1.452(0.042)$ & $
					1.501(0.065)$ & $1.520(0.066)$ & $0.355(0.012)$  &\multicolumn{1}{c}{$-$} \\
					\multicolumn{1}{c|}{SGLV}&  $ 0.261(0.011)$ & $0.779(0.029)$ & $0.852(0.028)$ & $
					1.174(0.054)$ & $1.397(0.062)$ & $0.239(0.009)$ &$4.497(2.950)\times 10^{-6}$ \\
					\cline{1-8}
					&   \multicolumn{1}{c}{$a_{41}$} & \multicolumn{1}{c}{$a_{42}$} & \multicolumn{1}{c}{$a_{43}$}&  \multicolumn{1}{c}{$a_{44}$} & \multicolumn{1}{c|}{$a_{45}$} & \multicolumn{1}{c}{$r_{4}$} &\multicolumn{1}{c}{$\sigma^{2}_{4}$} \\
					\multicolumn{1}{c|}{GLV}&  $ 0.355(0.013)$ & $0.528(0.025)$ & $0.310(0.013)$ &
					$5.643(0.140)$ & $1.080(0.047)$ & $0.170(0.007)$  &\multicolumn{1}{c}{$-$} \\
					\multicolumn{1}{c|}{SGLV}&  $ 0.267(0.011)$ & $0.450(0.021)$ & $0.292(0.013)$ &
					$3.164(0.098)$ & $0.901(0.041)$ & $0.135(0.006)$ &$5.274(3.457)\times 10^{-6}$ \\
					\cline{1-8}
					&  \multicolumn{1}{c}{$a_{51}$} & \multicolumn{1}{c}{$a_{52}$} & \multicolumn{1}{c}{$a_{53}$} &  \multicolumn{1}{c}{$a_{54}$} & \multicolumn{1}{c|}{$a_{55}$} & \multicolumn{1}{c}{$r_{5}$} & \multicolumn{1}{c}{$\sigma^{2}_{5}$}\\
					\multicolumn{1}{c|}{GLV}&  $ 0.349(0.014)$ & $0.995(0.037)$ & $0.970(0.035)$ &
					$2.148(0.088)$ & $2.821(0.101)$ & $0.320(0.011)$ &\multicolumn{1}{c}{$-$} \\
					\multicolumn{1}{c|}{SGLV}&   $ 0.268(0.011)$ & $0.700(0.028)$ & $0.676(0.025)$ &
					$1.604(0.068)$ & $1.832(0.071)$ & $0.217(0.008)$ &$2.995(2.540)\times 10^{-6}$ \\
					\cline{1-8}
					\hline
			\end{tabular}}
		}
	\end{center}
\end{table}
\renewcommand\arraystretch{0.6}
\begin{table}[!htbp]
	\caption{Simulation results for Case~2, $n=300 $ based on 1000
		simulations: mean squared errors (standard error) of estimates of $\{a_{kl}, 1\leq k\leq 5, 1\leq l\leq 5\}$ and $\{r_k, 1\leq k\leq 5\}$.}\label{Case2-300}\begin{center}
		\resizebox{\textwidth}{28mm}{
			\setlength{\tabcolsep}{0.5mm}{
				\begin{tabular}{r|rrrrr|rr}
					\hline
					\hline
					\multicolumn{1}{c|}{Method}&\multicolumn{5}{c}{Mean squared error (standard error) }&\multicolumn{1}{r}{}& \multicolumn{1}{r}{}\\
					\hline
					&  \multicolumn{1}{c}{$ a_{11}$} & \multicolumn{1}{c}{$a_{12}$} & \multicolumn{1}{c}{$ a_{13}$} & \multicolumn{1}{c}{$a_{14}$} & \multicolumn{1}{c|}{$a_{15}$} & \multicolumn{1}{c}{$r_{1}$} &\multicolumn{1}{c}{$ \sigma^{2}_{1}$ } \\
					\multicolumn{1}{c|}{GLV}&   $6.939(0.592)$ & $2.890(0.137) $ & $1.432(0.064)$ & $
					5.648(0.282)$ & $4.817(0.216)$ & $0.695(0.026)$ &\multicolumn{1}{c}{$-$} \\
					\multicolumn{1}{c|}{SGLV}& $6.221(0.455)$ & $2.153(0.104) $ & $1.055(0.048)$ & $
					4.327(0.205)$ & $3.793(0.175)$ & $0.233(0.011)$&$0.007(0.022)$\\
					\cline{1-8}
					&  \multicolumn{1}{c}{$a_{21}$} & \multicolumn{1}{c}{$a_{22}$} & \multicolumn{1}{c}{$a_{23}$} & \multicolumn{1}{c}{$ a_{24}$} & \multicolumn{1}{c|}{$a_{25}$} &\multicolumn{1}{c}{ $r_{2}$} &\multicolumn{1}{c}{$\sigma^{2}_{2}$} \\
					\multicolumn{1}{c|}{GLV}&  $3.811(0.290)$ & $2.488(0.105) $ & $1.632(0.066)$ & $
					4.949(0.211)$ & $4.441(0.266)$ & $0.744(0.024)$ &\multicolumn{1}{c}{$-$} \\
					\multicolumn{1}{c|}{SGLV}&  $ 3.481(0.251)$ & $2.002(0.100) $ & $1.239(0.050)$ & $
					3.921(0.170)$ & $3.618(0.220)$&$0.234(0.011)$ &$0.021(0.023)$ \\
					\cline{1-8}
					&  \multicolumn{1}{c}{$ a_{31}$} & \multicolumn{1}{c}{$a_{32}$} & \multicolumn{1}{c}{$a_{33}$} & \multicolumn{1}{c}{$a_{34}$} & \multicolumn{1}{c|}{$a_{35}$} & \multicolumn{1}{c}{$r_{3}$} &\multicolumn{1}{c}{$\sigma^{2}_{3}$} \\
					\multicolumn{1}{c|}{GLV}&  $ 3.692(0.223)$ & $2.671(0.117)$ & $1.501(0.055)$ &
					$ 3.878(0.198)$ & $3.808(0.185)$ & $0.838(0.024)$ &\multicolumn{1}{c}{$-$} \\
					\multicolumn{1}{c|}{SGLV}&  $3.059(0.186)$ & $2.060(0.092)$ & $1.018(0.043)$ &
					$ 3.123(0.148)$ & $3.271(0.156)$ & $0.253(0.010)$&$0.035(0.030)$ \\
					\cline{1-8}
					&   \multicolumn{1}{c}{$a_{41}$} & \multicolumn{1}{c}{$a_{42}$} & \multicolumn{1}{c}{$a_{43}$}&  \multicolumn{1}{c}{$a_{44}$} & \multicolumn{1}{c|}{$a_{45}$} & \multicolumn{1}{c}{$r_{4}$} &\multicolumn{1}{c}{$\sigma^{2}_{4}$} \\
					\multicolumn{1}{c|}{GLV}&  $3.969(0.272)$ & $1.866(0.090)$ & $0.876(0.041)$ &
					$ 5.498(0.212)$ & $3.759(0.198)$ & $0.502(0.017)$&\multicolumn{1}{c}{$-$} \\
					\multicolumn{1}{c|}{SGLV}&  $3.205(0.223)$ & $1.715(0.084)$ & $0.792(0.037)$ &
					$ 3.949(0.168)$ & $3.456(0.173)$ & $0.170(0.008)$ &$0.033(0.028)$ \\
					\cline{1-8}
					&  \multicolumn{1}{c}{$a_{51}$} & \multicolumn{1}{c}{$a_{52}$} & \multicolumn{1}{c}{$a_{53}$} &  \multicolumn{1}{c}{$a_{54}$} & \multicolumn{1}{c|}{$a_{55}$} & \multicolumn{1}{c}{$r_{5}$} & \multicolumn{1}{c}{$\sigma^{2}_{5}$}\\
					\multicolumn{1}{c|}{GLV}&  $ 4.442(0.316)$ & $2.310(0.112)$ & $1.546(0.067)$ &
					$ 4.705(0.211)$ & $4.723(0.199)$ & $0.802(0.025)$ &\multicolumn{1}{c}{$-$} \\
					\multicolumn{1}{c|}{SGLV}&   $ 3.376(0.216)$ & $1.874(0.092)$ & $1.123(0.049)$ &
					$ 3.645(0.164)$ & $3.784(0.192)$ & $0.239(0.010)$ &$0.023(0.021)$ \\
					\cline{1-8}
					\hline
		\end{tabular}}}
	\end{center}
\end{table}
\renewcommand\arraystretch{0.6}
\begin{table}[!htbp]
	\caption{Simulation results for Case~2, $n=500 $ based on 1000
		simulations: mean squared error(standard error) of estimates of $\{a_{kl}, 1\leq k\leq 5, 1\leq l\leq 5\}$ and $\{r_k, 1\leq k\leq 5\}$.}\label{Case2-500}\begin{center}
		\resizebox{\textwidth}{28mm}{
			\setlength{\tabcolsep}{0.5mm}{
				\begin{tabular}{r|rrrrr|rr}
					\hline
					\hline
					\multicolumn{1}{c|}{Method}&\multicolumn{5}{c}{mean squared error(standard error) }&\multicolumn{1}{r}{}& \multicolumn{1}{r}{}\\
					\hline
					&  \multicolumn{1}{c}{$ a_{11}$} & \multicolumn{1}{c}{$a_{12}$} & \multicolumn{1}{c}{$ a_{13}$} & \multicolumn{1}{c}{$a_{14}$} & \multicolumn{1}{c|}{$a_{15}$} & \multicolumn{1}{c}{$r_{1}$} &\multicolumn{1}{c}{$ \sigma^{2}_{1}$ } \\
					\multicolumn{1}{c|}{GLV}&   $6.166(0.677)$ & $1.840(0.082) $ & $1.049(0.044)$ & $
					3.157(0.144)$ & $2.690(0.117)$ & $0.596(0.018)$ &\multicolumn{1}{c}{$-$} \\
					\multicolumn{1}{c|}{SGLV}& $5.556(0.509)$ & $1.377(0.059) $ & $0.697(0.031)$ & $
					2.380(0.112)$ & $2.028(0.091)$ & $0.140(0.006)$&$0.004(0.020)$\\
					\cline{1-8}
					&  \multicolumn{1}{c}{$a_{21}$} & \multicolumn{1}{c}{$a_{22}$} & \multicolumn{1}{c}{$a_{23}$} & \multicolumn{1}{c}{$ a_{24}$} & \multicolumn{1}{c|}{$a_{25}$} &\multicolumn{1}{c}{ $r_{2}$} &\multicolumn{1}{c}{$\sigma^{2}_{2}$} \\
					\multicolumn{1}{c|}{GLV}&  $3.196(0.222)$ & $1.768(0.065) $ & $1.056(0.041)$ & $
					3.523(0.151)$ & $2.526(0.114)$ & $0.613(0.017)$ &\multicolumn{1}{c}{$-$} \\
					\multicolumn{1}{c|}{SGLV}&  $ 2.714(0.181)$ & $1.208(0.049) $ & $0.716(0.030)$ & $
					2.523(0.115)$ & $2.024(0.094)$&$0.140(0.006)$ &$0.017(0.017)$ \\
					\cline{1-8}
					&  \multicolumn{1}{c}{$ a_{31}$} & \multicolumn{1}{c}{$a_{32}$} & \multicolumn{1}{c}{$a_{33}$} & \multicolumn{1}{c}{$a_{34}$} & \multicolumn{1}{c|}{$a_{35}$} & \multicolumn{1}{c}{$r_{3}$} &\multicolumn{1}{c}{$\sigma^{2}_{3}$} \\
					\multicolumn{1}{c|}{GLV}&  $ 3.517(0.236)$ & $1.706(0.069)$ & $1.546(0.045)$ &
					$ 2.111(0.105)$ & $2.172(0.099)$ & $0.874(0.019)$ &\multicolumn{1}{c}{$-$} \\
					\multicolumn{1}{c|}{SGLV}&  $2.921(0.195)$ & $1.274(0.054)$ & $0.950(0.031)$ &
					$ 1.879(0.089)$ & $1.817(0.081)$ & $0.223(0.008)$&$0.030(0.023)$ \\
					\cline{1-8}
					&   \multicolumn{1}{c}{$a_{41}$} & \multicolumn{1}{c}{$a_{42}$} & \multicolumn{1}{c}{$a_{43}$}&  \multicolumn{1}{c}{$a_{44}$} & \multicolumn{1}{c|}{$a_{45}$} & \multicolumn{1}{c}{$r_{4}$} &\multicolumn{1}{c}{$\sigma^{2}_{4}$} \\
					\multicolumn{1}{c|}{GLV}&  $3.196(0.211)$ & $1.055(0.052)$ & $0.468(0.022)$ &
					$ 4.900(0.156)$ & $2.072(0.108)$ & $0.456(0.013)$&\multicolumn{1}{c}{$-$} \\
					\multicolumn{1}{c|}{SGLV}&  $2.871(0.207)$ & $0.953(0.047)$ & $0.417(0.018)$ &
					$ 3.028(0.108)$ & $1.856(0.100)$ & $0.114(0.005)$ &$0.029(0.022)$ \\
					\cline{1-8}
					&  \multicolumn{1}{c}{$a_{51}$} & \multicolumn{1}{c}{$a_{52}$} & \multicolumn{1}{c}{$a_{53}$} &  \multicolumn{1}{c}{$a_{54}$} & \multicolumn{1}{c|}{$a_{55}$} & \multicolumn{1}{c}{$r_{5}$} & \multicolumn{1}{c}{$\sigma^{2}_{5}$}\\
					\multicolumn{1}{c|}{GLV}&  $ 3.381(0.259)$ & $1.750(0.073)$ & $1.168(0.043)$ &
					$ 2.951(0.126)$ & $3.916(0.134)$ & $0.835(0.020)$ &\multicolumn{1}{c}{$-$} \\
					\multicolumn{1}{c|}{SGLV}&   $2.662(0.194)$ & $1.314(0.057)$ & $0.859(0.035)$ &
					$ 2.374(0.103)$ & $2.585(0.102)$ & $0.201(0.008)$ &$0.021(0.019)$ \\
					\cline{1-8}
					\hline
		\end{tabular}}}
	\end{center}
\end{table}

\renewcommand\arraystretch{0.6}
\begin{table}[!htbp]
	\caption{Simulation results for Case~2, $n=1000 $ based on 1000
		simulations: mean squared errors (standard error) of estimates of $\{a_{kl}, 1\leq k\leq 5, 1\leq l\leq 5\}$ and $\{r_k, 1\leq k\leq 5\}$. }\label{Case2-delta0-1000}\begin{center}
		\resizebox{\textwidth}{28mm}{
			\setlength{\tabcolsep}{0.5mm}{
				\begin{tabular}{r|rrrrr|rr}
					\hline
					\hline
					\multicolumn{1}{c|}{Method}&\multicolumn{5}{c}{Mean squared errors (standard error)}&\multicolumn{1}{r}{}& \multicolumn{1}{r}{}\\
					\hline
					&  \multicolumn{1}{c}{$ a_{11}$} & \multicolumn{1}{c}{$a_{12}$} & \multicolumn{1}{c}{$ a_{13}$} & \multicolumn{1}{c}{$a_{14}$} & \multicolumn{1}{c|}{$a_{15}$} & \multicolumn{1}{c}{$r_{1}$} &\multicolumn{1}{c}{$ \sigma^{2}_{1}$ } \\
					\multicolumn{1}{c|}{GLV}&   $5.817(0.714)$ & $1.038(0.044) $ & $0.657(0.027)$ & $
					1.524(0.071)$ & $1.215(0.059)$ & $0.508(0.013)$ &\multicolumn{1}{c}{$-$} \\
					\multicolumn{1}{c|}{SGLV}& $4.634(0.431)$ & $0.735(0.033) $ & $0.445(0.019)$ & $
					1.083(0.050)$ & $0.941(0.045)$ & $0.077(0.003)$&$0.002(0.018)$\\
					\cline{1-8}
					&  \multicolumn{1}{c}{$a_{21}$} & \multicolumn{1}{c}{$a_{22}$} & \multicolumn{1}{c}{$a_{23}$} & \multicolumn{1}{c}{$ a_{24}$} & \multicolumn{1}{c|}{$a_{25}$} &\multicolumn{1}{c}{ $r_{2}$} &\multicolumn{1}{c}{$\sigma^{2}_{2}$} \\
					\multicolumn{1}{c|}{GLV}&  $3.279(0.360)$ & $1.684(0.050) $ & $0.721(0.023)$ & $
					2.109(0.074)$ & $1.162(0.054)$ & $0.598(0.011)$ &\multicolumn{1}{c}{$-$} \\
					\multicolumn{1}{c|}{SGLV}&  $ 2.631(0.183)$ & $1.020(0.039) $ & $0.493(0.034)$ & $
					1.480(0.047)$ & $0.895(0.046)$&$0.098(0.012)$ &$0.013(0.014)$ \\
					\cline{1-8}
					&  \multicolumn{1}{c}{$ a_{31}$} & \multicolumn{1}{c}{$a_{32}$} & \multicolumn{1}{c}{$a_{33}$} & \multicolumn{1}{c}{$a_{34}$} & \multicolumn{1}{c|}{$a_{35}$} & \multicolumn{1}{c}{$r_{3}$} &\multicolumn{1}{c}{$\sigma^{2}_{3}$} \\
					\multicolumn{1}{c|}{GLV}&  $ 2.618(0.183)$ & $1.080(0.039)$ & $1.552(0.034)$ &
					$ 1.082(0.047)$ & $1.073(0.046)$ & $0.831(0.012)$ &\multicolumn{1}{c}{$-$} \\
					\multicolumn{1}{c|}{SGLV}&  $2.254(0.162)$ & $0.755(0.030)$ & $0.914(0.023)$ &
					$ 0.902(0.040)$ & $0.892(0.040)$ & $0.178(0.005)$&$0.029(0.019)$ \\
					\cline{1-8}
					&   \multicolumn{1}{c}{$a_{41}$} & \multicolumn{1}{c}{$a_{42}$} & \multicolumn{1}{c}{$a_{43}$}&  \multicolumn{1}{c}{$a_{44}$} & \multicolumn{1}{c|}{$a_{45}$} & \multicolumn{1}{c}{$r_{4}$} &\multicolumn{1}{c}{$\sigma^{2}_{4}$} \\
					\multicolumn{1}{c|}{GLV}&  $2.737(0.231)$ & $0.513(0.025)$ & $0.232(0.010)$ &
					$ 4.988(0.119)$ & $0.928(0.43)$ & $0.465(0.009)$&\multicolumn{1}{c}{$-$} \\
					\multicolumn{1}{c|}{SGLV}&  $2.383(0.199)$ & $0.458(0.022)$ & $0.193(0.009)$ &
					$ 2.900(0.082)$ & $0.781(0.036)$ & $0.074(0.003)$ &$0.027(0.019)$ \\
					\cline{1-8}
					&  \multicolumn{1}{c}{$a_{51}$} & \multicolumn{1}{c}{$a_{52}$} & \multicolumn{1}{c}{$a_{53}$} &  \multicolumn{1}{c}{$a_{54}$} & \multicolumn{1}{c|}{$a_{55}$} & \multicolumn{1}{c}{$r_{5}$} & \multicolumn{1}{c}{$\sigma^{2}_{5}$}\\
					\multicolumn{1}{c|}{GLV}&  $ 3.071(0.203)$ & $1.045(0.040)$ & $0.809(0.025)$ &
					$ 1.740(0.072)$ & $3.574(0.099)$ & $0.799(0.013)$ &\multicolumn{1}{c}{$-$} \\
					\multicolumn{1}{c|}{SGLV}&   $2.583(0.191)$ & $0.771(0.31)$ & $0.549(0.018)$ &
					$ 1.299(0.056)$ & $2.166(0.070)$ & $0.161(0.005)$ &$0.018(0.014)$ \\
					\cline{1-8}
					\hline
		\end{tabular}}}
	\end{center}
\end{table}

\renewcommand\arraystretch{0.6}
\begin{table}[htbp]
	\begin{center}
		\caption{5\% Confidence intervals for estimated
                  interaction coefficients for the five
			most abundant families in person~2$'$s gut using our approach. The significantly
			non-zero interaction coefficients are highlighted in blue (for negative values).}\label{GutFamilyInterval5}
		\resizebox{\textwidth}{15mm}{
			\setlength{\tabcolsep}{0.5mm}{
				\begin{tabular}{r|rrrrr}
					\hline
					\hline
					\diagbox{~}{$a_{kp}$}{~~~~~~~~~~~} & \multicolumn{1}{c}{B.} & \multicolumn{1}{c}{R.} & \multicolumn{1}{c}{L.} & \multicolumn{1}{c}{P.} & \multicolumn{1}{c}{B$'$.}  \\
					\hline
					\multicolumn{1}{l|}{Bacteroidaceae}       & $(-2.575,-1.420)$ & $\cellcolor{blue!30!white!70}{(-3.677,-0.747)}$&
					$\cellcolor{blue!30!white!70}{(-3.842,-0.228)}$ & $(-0.143,5.568)$ & $(-4.361,4.713)$ \\
					\multicolumn{1}{l|}{Ruminococcaceae}      & $(-0.338,0.354)$ & $(-4.836,-3.079)$ &
					\cellcolor{blue!30!white!70}{$(-3.037,-0.871)$} & $\cellcolor{red!30!white!70}{(1.641,5.063)}$ &$(-3.321,2.118)$  \\
					\multicolumn{1}{l|}{Lachnospiraceae}      & $(-0.884,0.155)$ & $\cellcolor{blue!30!white!70}{(-2.951,-0.314)}$ &
					$(-7.107,-3.854)$ & $ (-1.511,3.628)$ & $(-4.688,3.479)$  \\
					\multicolumn{1}{l|}{Porphyromonadaceae}   & $(-0.742,0.125)$ & $(-1.064,1.136)$ &
					$(-1.140,1.374)$ & $(-12.515,-8.228)$ & $(-2.996,3.817)$   \\
					\multicolumn{1}{l|}{Bacteroidales (unsp.)}& $\cellcolor{blue!30!white!70}{(-1.959,-0.686)}$ &
					$(-3.199,0.030)$ & $\cellcolor{blue!30!white!70}{(-5.313,-1.330)}$ & $(-1.131,5.162)$ & $(-13.984,-3.983)$ \\
					\hline
					\hline
		\end{tabular}}}
	\end{center}
\end{table}

\renewcommand\arraystretch{0.6}
\begin{table}[htbp]
	\begin{center}
		\caption{Estimated growth rates, interaction coefficients and
			diffusion coefficients for the five most abundant families in Person~2's
			gut under the deterministic GLV model. The significantly non-zero interaction
			coefficients (at the 5\% significance level) are highlighted in blue.}\label{GutFamilyTableGLV}
		\setlength{\tabcolsep}{1.5mm}{
			\begin{tabular}{r|rrrrr|rr}
				\hline
				\hline
				\diagbox{~}{$a_{kp}$}{~~~~~~~~~~~} & \multicolumn{1}{c}{B.} & \multicolumn{1}{c}{R.} & \multicolumn{1}{c}{L.} & \multicolumn{1}{c}{P.} & \multicolumn{1}{c}{B$'$.} & \multicolumn{1}{c}{$r_i$} & \multicolumn{1}{c}{$\sigma^{2}_i$} \\
				\hline
				\multicolumn{1}{l|}{Bacteroidaceae}       & $ -1.130$ & $-0.406$& $-0.902$ & $1.875$ & $1.174$ &  0.497 &-  \\
				\multicolumn{1}{l|}{Ruminococcaceae}      & $0.042$ & $ -3.068$ & \cellcolor{blue!30!white!70}{$-1.198$} & $1.207$ &$-0.081$ & 0.510 & -  \\
				\multicolumn{1}{l|}{Lachnospiraceae}      & $-0.018$ & $-0.930$ & $ -3.946$ & $ -1.021$ & $-0.223$ &  0.703 & - \\
				\multicolumn{1}{l|}{Porphyromonadaceae}   & $-0.225$ & $-0.057$ & $-0.328$ & $-8.630$ & $0.281$ &  0.572 & -  \\
				\multicolumn{1}{l|}{Bacteroidales (unsp.)}& \cellcolor{blue!30!white!70}{$-0.456$} &  $0.127$ & $-2.090$ & $1.431$ & $-8.072$ &  0.659 & - \\
				\hline
				\hline
		\end{tabular}}
	\end{center}
\end{table}
\renewcommand\arraystretch{0.6}
\begin{table}[htbp]
	\begin{center}
		\caption{Bootstrap 5\% Confidence intervals for deterministic
			GLV model estimates of interaction
			coefficients and growth rates for the five most abundant families in
			Person~2's gut. The significantly non-zero inter-species
			interaction coefficients are highlighted in
			blue.}\label{GutFamilyGLVConfidence} \resizebox{\textwidth}{15mm}{
			\setlength{\tabcolsep}{0.5mm}{
				\begin{tabular}{r|rrrrr}
					\hline
					\hline
					\diagbox{~}{$a_{kp}$}{~~~~~~~~~~~} & \multicolumn{1}{c}{B.} & \multicolumn{1}{c}{R.} & \multicolumn{1}{c}{L.} & \multicolumn{1}{c}{P.} & \multicolumn{1}{c}{B$'$.}  \\
					\hline
					\multicolumn{1}{l|}{Bacteroidaceae}       & $(-2.614,-0.707)$ & $(-3.941,1.087)$&
					$(-4.775,0.987)$ & $(-6.062,9.060)$ & $(-3.013,4.297)$ \\
					\multicolumn{1}{l|}{Ruminococcaceae}      & $(-0.597,0.707)$ & $(-5.068,-2.105)$ &
					\cellcolor{blue!30!white!70}{$(-3.138,-0.056)$} & $(-2.540,5.628)$ &$(-3.972,2.684)$  \\
					\multicolumn{1}{l|}{Lachnospiraceae}      & $(-0.914,0.426)$ & $(-3.713,0.343)$ &
					$(-7.329,-2.343)$ & $ (-6.414,4.292)$ & $(-3.378,2.603)$  \\
					\multicolumn{1}{l|}{Porphyromonadaceae}   & $(-0.830,0.461)$ & $(-1.324,2.109)$ &
					$(-3.172,2.311)$ & $(-12.779,-5.178)$ & $(-4.429,4.612)$   \\
					\multicolumn{1}{l|}{Bacteroidales (unsp.)}& \cellcolor{blue!30!white!70}{$(-2.064,-0.022)$} &
					$(-3.579,1.602)$ & $(-5.997,0.006)$ & $(-6.647,8.029)$ & $(-13.143,-4.039)$ \\
					\hline
					\hline
		\end{tabular}}}
	\end{center}
\end{table}

\renewcommand\arraystretch{0.6}
\begin{figure}[htbp]
	\centering
	\caption{ (a) Relative abundance of the family Bacteroidaceae over time; (b) Logarithm of the temporal data in (a) and the fitting curve (blue line) from the SGLV model. }\label{AbundancePlot}
	\subfigure[Proportion]{
		\label{level.sub.1}
		\includegraphics[scale=0.38]{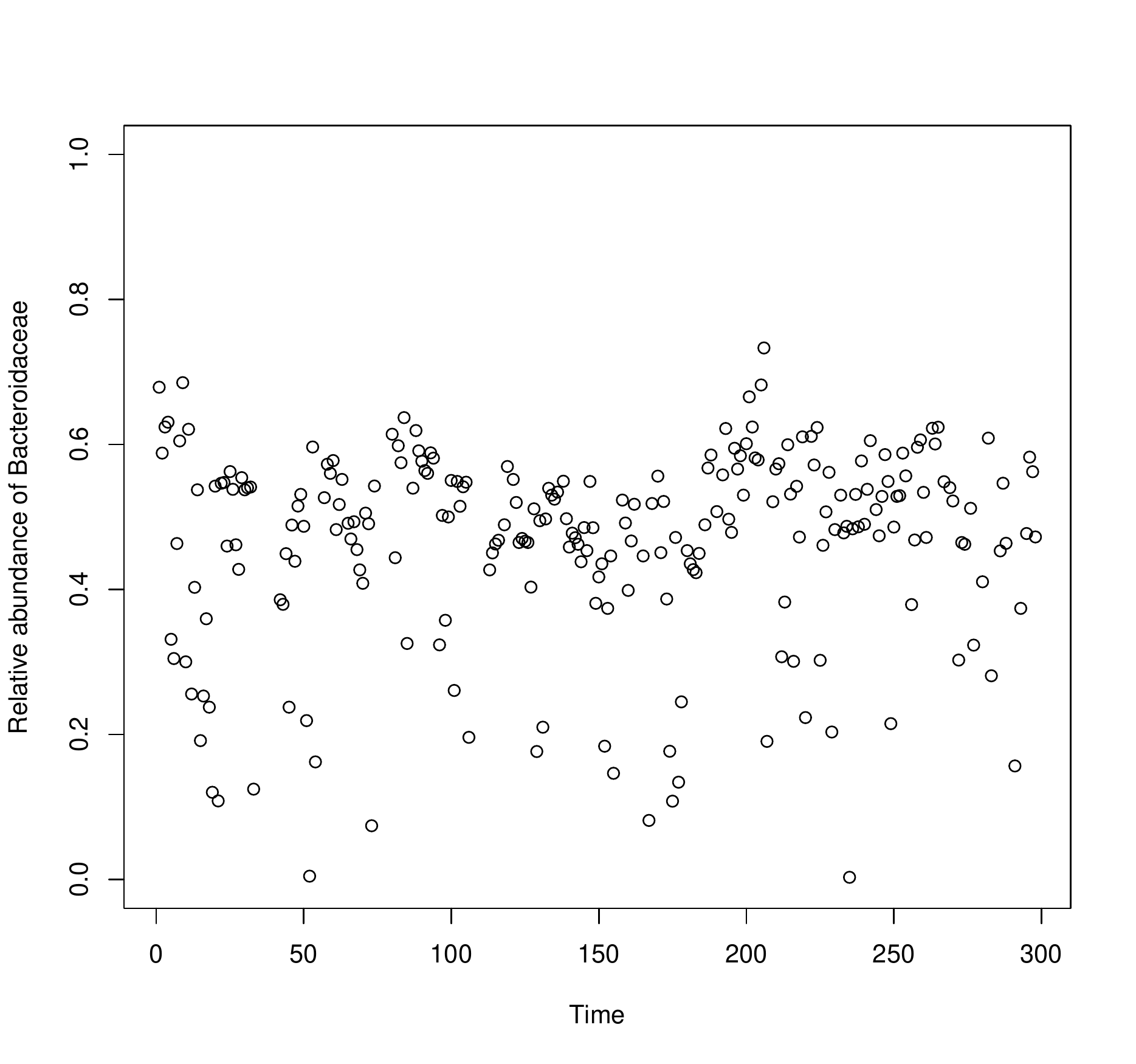}}
	\subfigure[log-proportion with one-step predictions (in blue)]{
		\label{logData}
		\includegraphics[width=0.48\linewidth]{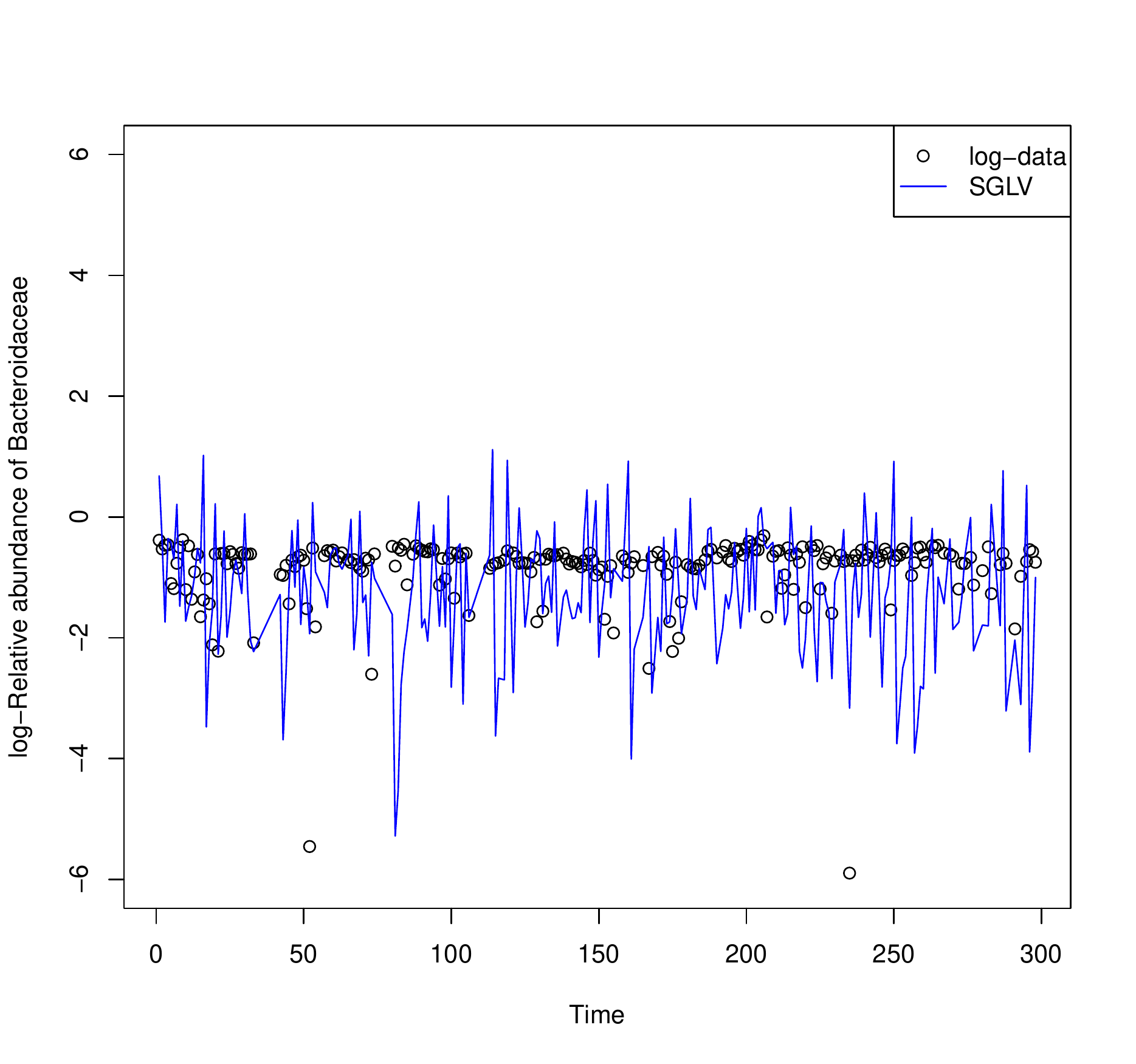}	}
	\label{level}
\end{figure}

\clearpage
\section{Proof of Propositions}\label{AppendProof}
We first show the existence and uniqueness of solutions to our stochastic differential equation model~\eqref{1.1}.
We first recall the assumptions needed:
\begin{assumption}
	\label{SupplementaryassumptionA1}
	The initial value $x(0)=(x_{1}(0),\ldots,x_{N}(0))^{\T}\in \mathbb{R}^{N}_{+}$, $r_{k}-\sigma^{2}_{k}/2>0,\sigma_{k}>0$,  and $A=(a_{kl})_{N\times N}$ is  non-positive definite.
\end{assumption}
\begin{assumption}
	\label{SupplementaryassumptionA2} For some $\phi\geq 4$, the elements of $A$ satisfy
	$a_{kk}+\phi(\phi+1)^{-1}\sum_{l=1}^N(a_{kl}\lor
	0)+(\phi+1)^{-1}\sum_{l=1}^N(a_{lk}\lor 0)<0.$
\end{assumption}
\begin{assumption}
	\label{SupplementaryassumptionA3}
	Each element of the vector $\tilde{x}=(\tilde{x}_{1},...,\tilde{x}_{N})^{\T}=-A^{-1}\left(r-\sigma^{2}/2\right)$ is positive, i.e.,  $\tilde{x}_{k}>0$ for $k=1,...,N$.
\end{assumption}
\begin{assumption}
	\label{SupplementaryassumptionA4}
	There exist positive constants $c_1$,$\ldots$, $c_N$ such that
	for $k=1,...,N$,
	$$\sum^{N}_{i=1}c_{i}\sigma^{2}_{i}\tilde{x}_{i} < -\left[2c_{k}a_{kk}
	+\sum_{l\neq k}\{c_{k}|a_{kl}|
	+c_{l}|a_{lk}|\}   \right]\tilde{x}^{2}_{k}.$$
\end{assumption}
\begin{proposition}
	For $t\geq0$, under Assumption~\ref{assumptionA1}, there is a unique solution $x(t)$ to SGLV model~\eqref{1.1} and  $x(t)\in \mathbb{R}^{N}_{+}$ almost surely.
\end{proposition}

\begin{proof}
	Since the coefficients of the equation are locally Lipschitz
	continuous, for any given initial value $x(0)\in \mathbb{R}_{+}$ there
	is a unique local solution $x(t)$ on $t \in [0,\tau_{e})$, where
	$\tau_{e}$ is the explosion time. To show that the solution is
	global, it is equivalent to show that $\tau_{e}=+\infty$ a.s.
	Define a $C^{2}$-function $V:\mathbb{R}^{N}_{+} \rightarrow \mathbb{R}_{\geq 0}$ by
	\begin{eqnarray}
		V(x(t))&=& \sum^{N}_{k=1}\{x_{k}(t)-1-\log (x_{k}(t))\}.\nonumber
	\end{eqnarray}
	Note that $V(x(t))$ is non-negative because $x_{k}(t)-1-\log
	(x_{k}(t)) \geq 0$ for all $x_{k}(t) > 0,k=1,...,N$. Furthermore,
	$V(x(t))\rightarrow\infty$ for any explosion time of $x(t)$, so
	$\tau_e$ is the explosion time of $V(x(t))$.
	
	We choose a sufficiently large non-negative number $\eta_{0}$ such
	that $V(x(0))\leq\eta_{0}$. For each $\eta\geq \eta_{0}$, we can
	define the stopping time
	\begin{eqnarray*}
		\tau_{\eta} &=&
		\inf\{t\in[0,\tau_{e})|V(x(t))\geq\eta\}
	\end{eqnarray*}
	Clearly, $\tau_{\eta}$ is increasing.  Set
	$\tau_{\infty}=\lim_{\eta\rightarrow\infty}\tau_{\eta}$. Since
	$\tau_{\infty}\leq \tau_{e}$, it is sufficient to prove
	$\tau_{\infty}=\infty$ a.s. This is equivalent to showing for any
	$T<\infty$ that $P(\tau_\eta \leq T)\rightarrow 0$ as
	$\eta\rightarrow\infty$.
	
	By It\^{o}'s formula, we get
	\begin{eqnarray}\label{last1}
		\mathrm{d}V(x(t))& =& F_{1}(x(t))dt
		+\sum^{N}_{k=1}\sigma_{k}\{x_{k}(t)-1\}dB_{k}(t),
	\end{eqnarray}
	where
	\begin{eqnarray}
		F_{1}(x(t)) &:=&
		\sum^{N}_{k=1}\{r_{k}x_{k}(t)-(r_{k}-\sigma^{2}_{k}/2)\}+\sum^{N}_{k=1}\{x_{k}(t)-1\}\sum^{N}_{l=1}a_{kl}x_{l}(t). \nonumber
	\end{eqnarray}
	Because $A$ is a non-positive definite matrix,
	$\sum^{N}_{k=1}\sum^{N}_{l=1}a_{kl}x_{k}(t)x_{l}(t)\leq 0$ and we have
	\begin{eqnarray}
		F_{1}(x(t)) &\leq&
		\sum^{N}_{k=1}r_{k}x_{k}(t)-\sum^{N}_{k,l=1}a_{kl}x_{l}(t)
		-\sum^{N}_{k=1}(r_{k}-\sigma^{2}_{k}/2)\nonumber\\ & \leq & C
		\sum^{N}_{k=1}x_{k}(t)+NC,\nonumber
	\end{eqnarray}
	where $C=\max_{k=1,...,N}\{(r_{k}-\sum^{N}_{l=1}a_{lk})\vee
	(r_{k}-\sigma^{2}_{k}/2)\vee 0\}$. From the inequality $x_{k}(t)\leq
	2\{x_{k}(t)-1-\log(x_{k}(t))\}+2$, we have
	\begin{equation}
		F_{1}(x(t)) \leq 2CV(x(t))+3NC,\nonumber
	\end{equation}
	so, from Equation~\eqref{last1}
	\begin{eqnarray}\label{last2}
		\mathrm{d}V(x(t))& \leq & \{2CV(x(t))+3NC\}dt
		+\sum^{N}_{k=1}\sigma_{k}\{x_{k}(t)-1\}dB_{k}(t),\nonumber
	\end{eqnarray}
	and
	\begin{equation}\label{last3}
		\int^{T\wedge \tau_{\eta}}_{0} dV(x(t)) \leq \int^{T\wedge
			\tau_{\eta}}_{0}\{2CV(x(t))+3NC\}dt +\int^{T\wedge
			\tau_{\eta}}_{0}\sum^{N}_{k=1}\sigma_{k}\{x_{k}(t)-1\}dB_{k}(t).\nonumber
	\end{equation}
	Taking expectations on both sides,
	\begin{eqnarray}\label{last4}
		E\{V(x(T\wedge \tau_{\eta}))\} &\leq& V(x(0))+ 2C
		E\left\{\int^{T\wedge \tau_{\eta}}_{0}V(x(t))dt\right\}+3NCE(T\wedge
		\tau_{\eta})\nonumber\\ & \leq& V(x(0))+3NCT+ 2C
		E\left\{\int^{T\wedge \tau_{\eta}}_{0}V(x(t))dt\right\}.\nonumber
	\end{eqnarray}
	By Gr\"{o}nwall's inequality, we have
	\begin{eqnarray}\label{last5}
		E\{V(x(T\wedge \tau_{\eta}))\}& \leq& \{V(x(0))+3 NCT\}
		\exp(2CT).\nonumber
	\end{eqnarray}
	By definition of $\tau_{\eta}$, we have $V(x(\tau_\eta))=\eta$, so
	$$\{V(x(0))+3 NCT\}\exp(2CT)\geq E\{V(x(T\land \tau_{\eta}))\}\geq \eta P(\tau_\eta\leq T).$$
	The left side of this inequality is bounded, so we must have
	$P(\tau_\eta\leq T)\rightarrow 0$ as $\eta\rightarrow\infty$.
\end{proof}

\begin{proposition}
	Under Assumptions~\ref{assumptionA1} and \ref{assumptionA2}, there exist positive constants $C_{1}$ and $C_{2}$ such that for any initial value $x_{0}\in \mathbb{R}_{+}^{N}$,
	the solution of SGLV~model~\eqref{1.1} has the properties
	\begin{equation}\label{Supplementarylimsup}
		\limsup_{t\rightarrow+\infty}\sum^{N}_{k=1}E\{x^{\theta}_{k}(t)\}\leq C_{1},\quad
		\limsup_{t\rightarrow+\infty}\sum^{N}_{k=1}t^{-1}E\left\{\int^{t}_{0}x^{\%theta}_{k}(s)ds\right\}\leq C_{2},\nonumber
	\end{equation}
	for any $0\leqslant \theta\leq 4$.
\end{proposition}

\begin{proof}
	Since the existence of higher moments implies the existence of lower
	moments, it is sufficient to prove the results for some
	$\theta\geqslant 4$. We show that the result holds for $\theta=\phi$
	from Assumption~\ref{assumptionA2}.  Define a $C^{2}$-function
	$V:\mathbb{R}^{N}_{+}\rightarrow \mathbb{R}_{+}$ by $V(x(t))
	=\sum^{N}_{k=1}x^{\phi}_{k}(t)$, where
	$x(t)=(x_{1}(t),...,x_{N}(t))^{\T}$.  We apply It\^{o}'s formula to
	$e^{t}V(x(t))$ and get
	\begin{equation*}
		\mathrm{d}\{e^{t}V(x(t))\} = e^{t}F(x(t))\mathrm{d}t+e^{t}\sum^{N}_{k=1}\phi\sigma_{k}x^{\phi}_{k}(t)\mathrm{d}B_{k}(t),\nonumber
	\end{equation*}
	where
	\begin{eqnarray}
		F(x(t)):=V(x(t))+\phi\sum^{N}_{k=1}r_{k}x^{\phi}_{k}(t)+\phi\sum^{N}_{k,l=1}a_{kl}x^{\phi}_{k}(t)x_{l}(t)+\{\phi(\phi-1)/2\}\sum^{N}_{k=1}\sigma^{2}_{k}x^{\phi}_{k}(t).
	\end{eqnarray}
	Note that, setting $J_{k}=\{l\ne k|a_{kl}<0\}$, we have
	\begin{eqnarray}
		\sum^{N}_{k,l=1}a_{kl}x^{\phi}_{k}(t)x_{l}(t)
		&=& \sum^{N}_{k=1}a_{kk}x^{\phi+1}_{k}(t)+\sum^{N}_{k=1}\left(\sum_{l\in J_{k}}+\sum_{l\in J^{c}_{k}}\right)a_{kl}x^{\phi}_{k}(t)x_{l}(t)\nonumber\\
		&\leq  & \sum^{N}_{k=1}a_{kk}x^{\phi+1}_{k}(t)+\sum^{N}_{k=1}\sum_{l\in J^{c}_{k}}a_{kl}x^{\phi}_{k}(t)x_{l}(t)\nonumber\\
		&\leq  & \sum^{N}_{k=1}a_{kk}x^{\phi+1}_{k}(t)+\sum^{N}_{k=1}\sum_{l\in J^{c}_{k}}a_{kl}\left\{
		\phi(\phi+1)^{-1} x^{\phi+1}_{k}(t)+(\phi+1)^{-1} x^{\phi+1}_{l}(t)\right\}\nonumber\\
		&=  & \sum^{N}_{k=1}x^{\phi+1}_{k}(t)\left\{a_{kk}+\phi(\phi+1)^{-1}\sum^{N}_{l=1}(a_{kl}\vee 0)+(\phi+1)^{-1}\sum^{N}_{l=1}(a_{lk}\vee 0)\right\}.\nonumber
	\end{eqnarray}
	Since $a_{kk}+\phi(\phi+1)^{-1}\sum^{N}_{l=1}(a_{kl}\vee 0)+(\phi+1)^{-1}\sum^{N}_{l=1}(a_{lk}\vee 0)<0$,
	$F(x(t))$ is bounded. Let $ C_{1} := \sup_{x(t)\in \mathbb{R}^{N}_{+} }F(x(t))<+\infty$,
	and we have
	\begin{equation}\label{2.7}
		\mathrm{d}\{\exp(t)V(x(t))\}
		\leq C_{1}\exp(t)\mathrm{d}t+\exp(t)\sum^{N}_{k=1}\phi\sigma_{k}x^{\phi}_{k}(t)\mathrm{d}B_{k}(t).
	\end{equation}
	
	Let $\eta_{0}>V(x(0))$. For each  $\eta\geq \eta_{0}$, we define the stopping time
	$$  \tau_{\eta} = \inf\{t\geq 0| V(x(t))\geq\eta\}.$$
	It then follows from Equation~\eqref{2.7}
	\begin{equation*}
		E\left[\int^{t\wedge\tau_{\eta}}_{0}d\{\exp(s)V(x(s))\}\right] \leq E \left[\int^{t\wedge\tau_{\eta}}_{0}C_{1}\exp(s)ds\right]
		+E\left[\int^{t\wedge\tau_{\eta}}_{0}\exp(s)\sum^{N}_{k=1}\phi\sigma_{k}x^{\phi}_{k}(s)dB_{k}(s)\right],\nonumber
	\end{equation*}
	where $\int^{t\wedge\tau_{\eta}}_{0}\exp(s)\sum^{N}_{k=1}\phi\sigma_{k}x^{\phi}_{k}(s)dB_{k}(s)$ is a martingale with expectation equal to zero, so we have
	\begin{equation*}
		E\{\exp(t\wedge\tau_{\eta})V(x(t\wedge\tau_{\eta}))\} \leq  V(x_{0})+C_{1}E\left(\int^{t\wedge\tau_{\eta}}_{0}\exp(s)ds\right).\nonumber
	\end{equation*}
	By the proof of Proposition~\ref{pro2.1},
	$\tau_{\eta}\rightarrow\infty$ almost surely as $\eta\rightarrow+\infty$. So we have
	\begin{equation}\label{2.8}
		\exp(t)E\{V(x(t))\} \leq  V(x_{0})+C_{1}(\exp(t)-1),
	\end{equation}
	thus
	\begin{eqnarray}
		\limsup_{t\rightarrow+\infty} \sum^{N}_{k=1}E\{x^{\phi}_{k}(t)\} &\leq & C_{1}.\nonumber
	\end{eqnarray}
	Integrating both sides of equation (\ref{2.8}) divided by $\exp(t)$, we can get
	\begin{eqnarray}
		\limsup_{t\rightarrow+\infty}t^{-1}\int^{t}_{0}E\{V(x(s))\}ds &\leq & C_{2}.\nonumber
	\end{eqnarray}
	By Fubini's theorem, we obtain the result.
\end{proof}

The proof of Proposition~3, uses the following theorems from Chapter~4 of \cite{RZH:1980}.

{
	\let\temptheorem\thetheorem
	\renewcommand{\thetheorem}{4.1}
	
	\begin{theorem}[Stationarity]\label{2}
		If Assumptions~\ref{assumptionB.1}--\ref{assumptionB.2} hold, then the Markov process $X(t)$ has a unique
		stationary distribution $\mu(\cdot)$.
	\end{theorem}
	
	\renewcommand{\thetheorem}{4.2}
	
	\begin{theorem}[Ergodicity]\label{3}
		Suppose that Assumptions~\ref{assumptionB.1}--\ref{assumptionB.2} hold, and let $\mu$ be the
		stationary distribution of the process $X(t)$.
		Let $f(\cdot)$ be a function integrable with respect to the measure $\mu$. Then
		\begin{eqnarray}
			P_{x}\left\{\lim_{T\rightarrow+\infty}\frac{1}{T}\int^{T}_{0}f(X(t))dt=\int_{\mathbb{R}^{N}}f(x)\mu(dx) \right\}&=&1
		\end{eqnarray}
		for all $x\in \mathbb{R}^{N}$.
	\end{theorem}
	
	\let\thetheorem\temptheorem
	
}

The Assumptions~\ref{assumptionB.1}--\ref{assumptionB.2} are as follows: There exists a bounded open domain $U\subset \mathbb{R}^{N}$ with
regular boundary $\Gamma$, having the following properties:

{\let\tempassumption\theassumption
	\renewcommand{\theassumption}{B.1}
	
	\begin{assumption}
		\label{assumptionB.1}
		In the domain $U$ and some neighborhood thereof, the smallest
		eigenvalue of the diffusion matrix $\Sigma(x)$ is bounded away from zero.
	\end{assumption}
	
	\renewcommand{\theassumption}{B.2}
	
	\begin{assumption}
		\label{assumptionB.2}
		If $x\in \mathbb{R}^{N}\setminus U$, the mean time $\tau$ at which a path  issuing from $x$ reaches the set $U$ is finite, and $\sup_{x\in K}E_{x} \tau<\infty$ for every
		compact subset $K\subset\mathbb{R}^{N}$.
	\end{assumption}
	
	\let\theassumption\tempassumption
}

\begin{proposition}\label{Supplementarypro2.3}
	Under Assumptions~\ref{assumptionA1}, \ref{assumptionA3} and \ref{assumptionA4}, there is a stationary distribution for the solution of SGLV~model~\eqref{1.1}, and it has the ergodic property.
\end{proposition}

\begin{proof}
	To apply Theorems~\ref{2} and~\ref{3}, we first need to show that
	Assumptions~\ref{assumptionB.1}--\ref{assumptionB.2} hold.  For Assumption~\ref{assumptionB.1}, Since $\sigma_{k}>0$ in SGLV (3), the smallest eigenvalue of the diffusion matrix
	$\Sigma={\rm diag}(\sigma^{2}_{1},...,\sigma^{2}_{N})$ is bounded away from
	zero. To verify Assumption~\ref{assumptionB.2}, it is sufficient to show that there exists some
	neighborhood $U$ and a non-negative $C^{2}$-function $V(x(t))$ such
	that and for any $x(t)\in \mathbb{R}^{N}\setminus U$, $\mathcal{L}V$
	is negative (for details, refer to \cite{ZY:2007}), where
	\begin{equation}\label{Ldefinition}
		\mathcal{L}:=\sum^{N}_{k=1}x_{k}(t)\left\{r_{k}+\sum^{N}_{l=1}a_{kl}x_{l}(t)\right\}\partial(\cdot)\{\partial
		x_{k}(t)\}^{-1}+2^{-1}\sum^{N}_{k=1}\sigma^{2}_{k}\partial^{2}(\cdot)\{\partial
		x_{k}(t)\partial x_{l}(t)\}^{-1}.
	\end{equation}
	
	By Assumption~\ref{assumptionA3} , $\tilde{x}=(\tilde{x}_{1},...,\tilde{x}_{N})^{\T}$ is positive.
	We define a non-negative $C^{2}-$function
	$V^{\ast}(x(t))=\sum^{N}_{k=1}c_{k}\{x_{k}(t)-\tilde{x}_{k}-\tilde{x}_{k}\log(x_{k}(t)/\tilde{x}_{k})\}$ and will show that there exists a constant $C^{\ast}>0$ such that
	$$
	\mathcal{L}V^{\ast}(x(t))\leq -C^{\ast},
	$$
	where $\mathcal{L}$, defined by Equation~\eqref{Ldefinition},
	denotes the infinitesimal generator of stochastic differential equation~\eqref{1.1}.
	It is easy to see that $V^{\ast}(x(t))\rightarrow+\infty$ as $x(t)\rightarrow
	+\infty$ and $V^{\ast}(x(t))\rightarrow+\infty$ as $x(t)\rightarrow 0$.
	With some computation, we can get
	\begin{eqnarray}
		&&  \mathcal{L}V^{\ast}(x(t)) \nonumber\\
		&=& \sum^{N}_{k=1}c_{k}a_{kk}\{x_{k}(t)-\tilde{x}_{k}\}^{2}
		+\sum^{N}_{k=1}\sum^{N}_{l\neq k}c_{k}a_{kl}\{x_{k}(t)-\tilde{x}_{k}\}\{x_{l}(t)-
		\tilde{x}_{l}\}+\frac{1}{2}\sum^{N}_{k=1}c_{k}\sigma^{2}_{k}\tilde{x}_{k}.\nonumber
	\end{eqnarray}
	Note that
	\begin{eqnarray}
		&& \sum^{N}_{k=1}\sum^{N}_{l\neq k}c_{k}a_{kl}\{x_{k}(t)-\tilde{x}_{k}\}\{x_{l}(t)-
		\tilde{x}_{l}\}\nonumber\\
		&\leq &\frac{1}{2}\sum^{N}_{k=1}\sum^{N}_{l\neq k}c_{k}|a_{kl}|\big[\{x_{k}(t)-\tilde{x}_{k}\}^{2}+\{x_{l}(t)-
		\tilde{x}_{l}\}^{2}\big],\nonumber
	\end{eqnarray}
	so
	\begin{equation}\label{1.2.1}
		\mathcal{L}V^{\ast}(x(t))
		\leq  \frac{1}{2}\sum^{N}_{k=1}\left\{2c_{k}a_{kk}
		+\sum_{l\neq k}(c_{k}|a_{kl}|
		+c_{l}|a_{lk}|)
		\right\}\{x_{k}(t)-\tilde{x}_{k}\}^{2}+\frac{1}{2}\sum^{N}_{k=1}c_{k}\sigma^{2}_{k}\tilde{x}_{k}.\nonumber
	\end{equation}
	Assumption~\ref{assumptionA4} gives that,
	\begin{equation*}
		\sum^{N}_{k=1}c_{k}\sigma^{2}_{k}\tilde{x}_{k} < \min\left\{-\left\{2c_{k}a_{kk}
		+\sum_{l\neq k}(c_{k}|a_{kl}|
		+c_{l}|a_{lk}|)\right\}\tilde{x}^{2}_{k}, k=1,...,N\right\}.
	\end{equation*}
	We consider the ellipsoid
	\begin{equation*}\label{el}
		\mathbb{E}=\left\{x(t)\middle|\sum^{N}_{k=1}\left\{2c_{k}a_{kk}
		+\sum_{l\neq k}(c_{k}|a_{kl}|
		+c_{l}|a_{lk}|)\right\}\left\{x_{k}(t)-\tilde{x}_{k}\right\}^{2}+\sum^{N}_{k=1}c_{k}\sigma^{2}_{k}\tilde{x}_{k}\geq
		0\right\}.
	\end{equation*}
	It is easy to see that $\mathcal{L}V^{\ast}(x(t))>0$ in some subset of
	$\mathbb{E}$ (e.g.$\{x(t)=\tilde{x}(t)\}$)and
	$\mathcal{L}V^{\ast}(x(t))<0$ outside $\mathbb{E}$, so we can find a $\mathbb{D}\subseteq
	\mathbb{R}^{N}_{+}$ and a positive constant $C^{\ast}$, such that
	$\mathbb{D}\supset\mathbb{E}$ and $\mathcal{L}V^{\ast}(x(t))\leq
	-C^{\ast},~x(t)\in \mathbb{R}^{N}_{+}\setminus \bar{\mathbb{D}}$,
	where $\bar{\mathbb{D}}$ is the closure of $\mathbb{D}$, which implies the solution $x(t)$ is recurrent in the
	domain $\mathbb{D}$.
	By Theorem~\ref{2}, $x(t)$ has a stationary distribution, and by the
	Theorem~\ref{3}, it also has the ergodic property.
\end{proof}

\begin{corollary}
	Under Assumptions~\ref{assumptionA1}, \ref{assumptionA3} and \ref{assumptionA4}, for any Borel measurable function $f(\cdot):\mathbb{R}^{N}_{+}\rightarrow \mathbb{R}$,
	which is integrable
	with respect to the density, $\pi(\cdot)$, of the stationary distribution, the solution of SGLV~model \eqref{1.1} has the property,
	\begin{equation}\label{Supplementary2.11}
		\lim_{t\rightarrow+\infty}t^{-1}\int^{t}_{0}f(x(s))ds=\int_{\mathbb{R}^{N}_{+}}f(x)\pi(dx).\nonumber
	\end{equation}
\end{corollary}

\begin{proof}
	By Proposition~\ref{pro2.3},
	there exists a stationary density for the stochastic differential equation, denoted by $\pi(\cdot)$.
	The ergodic theory on the stationary distribution (Theorem 4.2 from \cite{RZH:1980}) gives that
	for any Borel measurable function $f(\cdot):\mathbb{R}^{N}_{+}\rightarrow \mathbb{R}$,
	which is integrable
	with respect to $\pi(\cdot)$,
	\begin{equation}
		\lim_{t\rightarrow+\infty}t^{-1}\int^{t}_{0}f(x(s))ds=\int_{\mathbb{R}^{N}_{+}}f(x)\pi(dx),\nonumber
	\end{equation}
	holds with probability 1.
\end{proof}

\section{Proof of Theorems}\label{AppendProofTheorem}
Recall that the main theorems depend on the following additional assumptions:

\setcounter{assumption}{4}

\begin{assumption}
	\label{SupplementaryassumptionA5}
	$T$ is fixed, and $\Delta_{\max}\rightarrow0 $.
\end{assumption}
\begin{assumption}
	\label{SupplementaryassumptionA6}
	(I) $T\rightarrow+\infty$ and $\Delta_{\max}\rightarrow0 $; (II) $T\rightarrow+\infty$ and
	$T\Delta_{\max}\rightarrow 0$.
\end{assumption}

Before proving Theorem \ref{FixedT},
we first prove the following key
lemmas:
{\let\templemma\thelemma
	\renewcommand{\thelemma}{S1}
	
	\begin{lemma}\label{lemma2}
		Suppose Assumptions~\ref{assumptionA1}, \ref{assumptionA2}, and \ref{assumptionA5}, hold. For fixed N and T, let $\Psi=\Theta \times \Phi$,
		and $\mathcal{K}\subset\Theta$ be a relatively compact set. Then for any $\tilde{\theta}=(\tilde{\vartheta},\sigma)\in \Psi$
		, and $r=0,1,2$.
		\begin{equation}\label{lemma7.1}
		\sup_{\vartheta\in\overline{\mathcal{K}}}|(\partial/\partial \vartheta)^{r}\ell_{n,T}(\vartheta)-(\partial/\partial \vartheta)^{r}\ell_{T}(\vartheta)|=O_{P_{\tilde{\theta}}}(\Delta^{1/2}_{\max}),n\rightarrow+\infty,
		\end{equation}
		where $\overline{\mathcal{K}}$ is the closure of $\mathcal{K}$.
	\end{lemma}
}

\begin{proof}
	Here we only consider the proof of the case $r=0$, because proofs of the cases $r=1$ and $r=2$ are similar to the case $r=0$.
	In order to show~\eqref{lemma7.1}, we only need to show that
	\begin{eqnarray}
		\lim_{\kappa\rightarrow+\infty}P\left\{\sup_{\vartheta\in\overline{\mathcal{K}}}\big|\ell_{n,T}(\vartheta)-\ell_{T}(\vartheta)\big|>\kappa \Delta^{1/2}_{\max}\right\}
		\leq \lim_{\kappa\rightarrow+\infty} \kappa^{-1}\Delta^{-1/2}_{\max} E\left\{\sup_{\vartheta\in\overline{\mathcal{K}}}\big|\ell_{n,T}(\vartheta)-\ell_{T}(\vartheta)\big|\right\}=0.\nonumber
	\end{eqnarray}
	We start by setting
	\begin{eqnarray}
		\ell_{n,T}(\vartheta)-\ell_{T}(\vartheta)
		&=&\sum^{N}_{k=1}\sigma^{-2}_{k}\sum^{n-1}_{i=0}\int^{t_{i+1}}_{t_{i}}
		\left(\sum^{N}_{l=1}a_{kl}\{x_{l}(t_{i})-x_{l}(t)\}\right)
		\left\{\tilde{R}_{k}+\sum^{N}_{m=1}\tilde{a}_{km}x_{m}(t)\right\}dt\nonumber\\
		&&+\sum^{N}_{k=1}\sigma^{-1}_{k}\sum^{n-1}_{i=0}\int^{t_{i+1}}_{t_{i}}\sum^{N}_{l=1}a_{kl}\{x_{l}(t_{i})
		-x_{l}(t)\}dB_{k}(t)\nonumber\\
		&&-\sum^{N}_{k=1}2^{-1}\sigma^{-2}_{k}\sum^{n-1}_{i=0}\int^{t_{i+1}}_{t_{i}}\big\{R_{k}+\sum^{N}_{l}a_{kl}x_{l}(t_{i})\big\}^{2}-\big\{R_{k}+\sum^{N}_{l}a_{kl}x_{l}(t)\big\}^{2}dt\nonumber\\
		&:= & \mathrm{I}_{1}+\mathrm{I}_{2}+\mathrm{I}_{3},\nonumber
	\end{eqnarray}
	where $\mathrm{I}_{1},\mathrm{I}_{2},\mathrm{I}_{3}$ are the three
	terms of the above equation.  It is therefore sufficient to prove that
	for some constant $C$, we have $E\left(\sup_{\vartheta\in
		\overline{\mathcal{K}}}|\mathrm{I}_{i}|\right)<C\Delta_{\max}^{-1/2}$,
	for $i=1,2,3$.
	
	\begin{eqnarray}
		E\sup_{\vartheta\in\bar{\mathcal{K}}}\mathrm{I}_{1} &\leq&
		\sum^{N}_{k=1}\sigma^{-2}_{k}\sum^{N}_{l=1}\left(\sup_{\vartheta\in\bar{\mathcal{K}}}|a_{kl}\tilde{R}_{k}|\right)\sum^{n-1}_{i=0}\int^{t_{i+1}}_{t_{i}}
		|\{x_{l}(t_{i})-x_{l}(t)\}|dt\nonumber\\
		&&+\sum^{N}_{k=1}\sum^{N}_{l,m=1}\sigma^{-2}_{k}\left(\sup_{\vartheta\in\bar{\mathcal{K}}}|a_{kl}\tilde{a}_{km}|\right)\sum^{n-1}_{i=0}\int^{t_{i+1}}_{t_{i}}
		|\{x_{l}(t_{i})-x_{l}(t)\}x_{m}(t)|dt,\nonumber\\
		E\sup_{\vartheta\in\bar{\mathcal{K}}}\mathrm{I}_{2} &\leq&
		\sum^{N}_{k=1}\sum^{N}_{l=1}\sigma^{-1}_{k}\left(\sup_{\vartheta\in\bar{\mathcal{K}}}|a_{kl}|\right)\sum^{n-1}_{i=0}E\left|\int^{t_{i+1}}_{t_{i}}\{x_{l}(t_{i})
		-x_{l}(t)\}dB_{k}(t)\right|\nonumber\\
		&\leq&
		\sum^{N}_{k=1}\sum^{N}_{l=1}\sigma^{-1}_{k}\left(\sup_{\vartheta\in\bar{\mathcal{K}}}|a_{kl}|\right)\sum^{n-1}_{i=0}
		\left(E\left[\int^{t_{i+1}}_{t_{i}}\{x_{l}(t_{i})
		-x_{l}(t)\}dB_{k}(t)\right]^{2}\right)^{1/2}\nonumber\\
		&=&
		\sum^{N}_{k=1}\sum^{N}_{l=1}\sigma^{-1}_{k}\left(\sup_{\vartheta\in\bar{\mathcal{K}}}|a_{kl}|\right)\sum^{n-1}_{i=0}
		\left(E\left[\int^{t_{i+1}}_{t_{i}}\{x_{l}(t_{i})
		-x_{l}(t)\}^{2}d(t)\right]\right)^{1/2},\nonumber\\
		E\sup_{\vartheta\in\bar{\mathcal{K}}}\mathrm{I}_{3} &\leq&
		\sum^{N}_{k=1}\sum^{N}_{l=1}\sigma^{-2}_{k}\left(\sup_{\vartheta\in\bar{\mathcal{K}}}
		|R_{k}a_{kl}|\right)\sum^{n-1}_{i=0}\int^{t_{i+1}}_{t_{i}}
		\left|\big\{x_{l}(t_{i})-x_{l}(t)\big\}\right|dt\nonumber\\
		&&+ \sum^{N}_{k=1} \sum^{N}_{l,m=1}\sigma^{-2}_{k}\left(\sup_{\vartheta\in\bar{\mathcal{K}}}
		|R_{k}a_{kl}a_{km}|\right)\sum^{n-1}_{i=0}\int^{t_{i+1}}_{t_{i}}
		|\{x_{l}(t_{i})-x_{l}(t)\}\{x_{m}(t_{i})-x_{m}(t)\}|dt.\nonumber
	\end{eqnarray}
	Since $\mathcal{K}$ is a relatively compact set, the terms,
	$\sup_{\vartheta \in \bar{\mathcal{K}}}|\tilde{R}_{k}a_{kl}|$,
	$\sup_{\vartheta \in \bar{\mathcal{K}}}|a_{kl}\tilde{a}_{km}|$,
	$\sup_{\vartheta \in \bar{\mathcal{K}}}|a_{kl}|$,
	$\sup_{\vartheta \in \bar{\mathcal{K}}}|R_{k}a_{kl}|$ and
	$\sup_{\vartheta \in \bar{\mathcal{K}}}|R_{k}a_{kl}a_{km}|$, are all bounded.
	By It\^{o}'s isometry and Jensen's inequality, we have
	\begin{eqnarray}
		E\left|\int^{t_{i+1}}_{t_{i}}\{x_{l}(t_{i})-x_{l}(t)\}x_{m}(t)dt\right|
		&\leq& \int^{t_{i+1}}_{t_{i}}E|\{x_{l}(t_{i})-x_{l}(t)\}x_{m}(t)|dt,\nonumber
	\end{eqnarray}
	and
	\begin{eqnarray}
		&&\left[E\left|\int^{t_{i+1}}_{t_{i}}\!\!\!\big\{x_{l}(t_{i})-x_{l}(t)\}dB_{k}(t)\right|\right]^{2}\nonumber\\
		&\leq& E\left[\int^{t_{i+1}}_{t_{i}}\{x_{l}(t_{i})-x_{l}(t)\}dB_{k}(t)\right]^{2}\nonumber\\
		&=&
		\int^{t_{i+1}}_{t_{i}}E\left(\{x_{l}(t_{i})-x_{l}(t)\}^{2}\right)dt\nonumber\\
		&=&
		\int^{t_{i+1}}_{t_{i}}E[\{x_{l}(t_{i})-x_{l}(t)\}x_{l}(t_{i})]dt-   \int^{t_{i+1}}_{t_{i}}E[\{x_{l}(t_{i})-x_{l}(t)\}x_{l}(t)]dt.\nonumber
	\end{eqnarray}
	It is enough to show that there exist three constants $C$, $C_{1}$ and $C_{2}$ such that
	\begin{eqnarray}\label{3.3.3.11}
		E|\{x_{l}(t_{i})-x_{l}(t)\}x_{m}(t)| &\leq & C_{1}\Delta_{\max}+C_{2}\Delta^{1/2}_{\max},
	\end{eqnarray}
	and
	\begin{eqnarray}\label{3.3.3.22}
		E\left(\{x_{l}(t_{i})-x_{l}(t)\}^{2}\right) &\leq & C\Delta_{i,t}.
	\end{eqnarray}
	We prove~\eqref{3.3.3.11} using Proposition~\ref{pro2.2}
	and It\^{o}'s isometry:
	\begin{eqnarray}
		&&E|\{x_{l}(t)-x_{l}(t_{i})\}x_{m}(t)|\nonumber\\
		&=& E\left|\int^{t}_{t_{i}}\left\{r_{k}+\sum^{N}_{l=1}a_{kl}x_{l}(s)\right\}x_{k}(s)x_{m}(t)ds+\int^{t}_{t_{i}}\sigma_{k}x_{k}(s)x_{m}(t)dB_{k}(s)\right|\nonumber\\
		&\leq&  E\left|\int^{t}_{t_{i}}\left\{r_{k}+\sum^{N}_{l=1}a_{kl}x_{l}(s)\right\}x_{k}(s)x_{m}(t)ds\right|
		+E\left|\int^{t}_{t_{i}}\sigma_{k}x_{k}(s)x_{m}(t)dB_{k}(s)\right|\nonumber\\
		&\leq &
		E\bigg|
		\int^{t}_{t_{i}}\left\{r_{k}+\sum^{N}_{l=1}a_{kl}x_{l}(s)\right\}x_{k}(s)x_{m}(t)ds\bigg|
		+\sigma_{k}\left[E\left\{\int^{t}_{t_{i}}x^{2}_{m}(t)x^{2}_{k}(s)ds\right\}\right]^{1/2}\nonumber\\
		&\leq &  C_{1}\Delta_{\max}+C_{2}\Delta^{1/2}_{\max}.
	\end{eqnarray}
	For~\eqref{3.3.3.22}, we have:
	\begin{eqnarray}
		E\left(\{x_{l}(t_{i})-x_{l}(t)\}^{2}\right)
		&=&E\left[\int^{t}_{t_{i}}\left\{r_{l}+\sum^{N}_{p=1}a_{lp}x_{p}(s)\right\}x_{l}(s)ds
		+\int^{t}_{t_{i}}\sigma_{l}dB_{l}(s)\right]^{2}\nonumber\\
		&\leq&2\Delta_{i,t} E\left[\int^{t}_{t_{i}}\left\{r_{l}+\sum^{N}_{p=1}a_{lp}x_{p}(s)\right\}^{2}x^{2}_{l}(s)ds\right]
		+2\int^{t}_{t_{i}}\sigma^{2}_{l}E\{x^{2}_{l}(s)\}ds\nonumber\\
		&\leq &C\Delta_{i,t}.\nonumber
	\end{eqnarray}
	Because $E\left(\sup_{\vartheta\in{\mathcal{K}}}|\mathrm{I}_{1}|\right)$,
	$E\left(\sup_{\vartheta\in{\mathcal{K}}}|\mathrm{I}_{2}|\right)$ and $E\left(\sup_{\vartheta\in{\mathcal{K}}}|\mathrm{I}_{3}|\right)$ are bounded by
	$C\Delta_{\max}^{-1/2}$ for some $C$, we have
	\begin{eqnarray}
		&&P_{\tilde{\theta}}\left(\Delta^{-1/2}_{\max}\sup_{\vartheta\in\mathcal{K}}|\ell_{n,T}-\ell_{T}(\vartheta)|\geq \kappa \right)\leq C\kappa^{-1}\nonumber
	\end{eqnarray}
	and
	\begin{eqnarray}
		\lim_{\kappa\rightarrow +\infty}\limsup_{n}P_{\tilde{\theta}}\bigg(\Delta^{-1/2}_{\max}\sup_{\vartheta\in\Theta}\left|T^{-1}\ell_{n,T}(\vartheta)-T^{-1}\ell_{T}(\vartheta)\right|\geq \kappa\bigg) &=& 0.\nonumber
	\end{eqnarray}
\end{proof}

\setcounter{theorem}{0}

\begin{theorem}
	Under Assumptions~\ref{assumptionA1}, \ref{assumptionA2}, and
	\ref{assumptionA5}, we have that, conditional on the maximum likelihood
	estimators lying in some compact parameter space $\mathcal{K}$,
	
	(i) $\|\hat{\vartheta}_{n,T}-\hat{\vartheta}_{T}\|=O_{p}\left(\Delta^{1/2}_{\max}\right)$.
	
	(ii) For any $k\in\{1,...,N\}$,
	$
	(n/2)^{1/2}\sigma^{-2}_{k}(\widehat{\sigma}^{2}_{k,n,T}-\sigma^{2}_{k}) \rightarrow N(0,1).\nonumber$
\end{theorem}

\begin{proof}
	We Taylor expand $\ell'_{T}(\hat{\vartheta}_{n,T})$ about
	$\hat{\vartheta}_{T}$. Because $\ell_T(\vartheta)$ is quadratic in
	$\vartheta$, the expansion is exact.
	\begin{eqnarray}
		\ell'_{T}(\hat{\vartheta}_{n,T}) &=& \ell'_{T}(\hat{\vartheta}_{T})+ \ell''_{T}(\hat{\vartheta}_{T})(\hat{\vartheta}_{n,T}-\hat{\vartheta}_{T}),\nonumber
	\end{eqnarray}
	since $\ell'_{T}(\hat{\vartheta}_{T})=0,\ell'_{n,T}(\hat{\vartheta}_{n,T})=0$,
	\begin{eqnarray}
		\|\hat{\vartheta}_{n,T}-\hat{\vartheta}_{T} \|&\leq&\|\{\ell''_{T}(\hat{\vartheta}_{T})\}^{-1}\|\cdot \|\ell'_{T}(\hat{\vartheta}_{n,T})-\ell'_{T}(\hat{\vartheta}_{T})\|\nonumber\\
		&=&\|\{\ell''_{T}(\hat{\vartheta}_{T})\}^{-1}\| \cdot
		\|\ell'_{T}(\hat{\vartheta}_{n,T})-\ell'_{n,T}(\hat{\vartheta}_{n,T})\|,\nonumber
	\end{eqnarray}
	where $\|\{\ell''_{T}(\hat{\vartheta}_{T})\}^{-1}\|=O_{p}(1)$ due to
	\begin{eqnarray}
		\partial^{2}\ell_{T}(\hat{\vartheta}_{T}))/(\partial a_{kl} \partial a_{km}) &=& -\int^{T}_{0}\exp\{u_{l}(t)+u_{m}(t)\}dt,\nonumber
	\end{eqnarray}
	and by Proposition~\ref{pro2.2},
	for any given $\epsilon$, there exists a constant $\kappa_{1}$ such that
	\begin{eqnarray}
		P( |\partial^{2}\ell_{T}(\hat{\vartheta}_{T})/(\partial a_{kl} \partial a_{km})|>\kappa_{1})  \leq \kappa^{-1}_{1}\int^{T}_{0}E\exp\{u_{l}(t)+u_{m}(t)\}dt\leq C\kappa^{-1}_{1}<\epsilon,\nonumber
	\end{eqnarray}
	further by Lemma~\ref{lemma2}, we have
	\begin{eqnarray}
		\|\hat{\vartheta}_{n,T}-\hat{\vartheta}_{T} \|&=&O_{p}(\Delta^{1/2}_{\max}).\nonumber
	\end{eqnarray}
	whenever $\hat{\vartheta}_{n,T}\in{\mathcal K}$. Next, we consider the asymptotic property of $\widehat{\sigma}^{2}_{k,n,T}$.
	\begin{eqnarray}
		&& n^{1/2}(\widehat{\sigma}^{2}_{k,n,T}-\sigma^{2}_{k})\nonumber\\
		&=& n^{1/2}\left[\widehat{\sigma}^{2}_{k,n,T}-n^{-1}\sigma^{2}_{k}\sum^{n-1}_{i=0}\{\Delta_{i}B_{k}(t)\}^{2}\Delta^{-1}_{i,t}\right]
		+2^{1/2}\sigma^{2}_{k}(2n)^{-1/2}\sum^{n-1}_{i=0}[\{\Delta_{i}B_{k}(t)\}^{2}-\Delta_{i,t}]\Delta^{-1}_{i,t}.\nonumber
	\end{eqnarray}
	Since $(2n)^{-1/2}\sum^{n-1}_{i=0}[\{\Delta_{i}B_{k}(t)\}^{2}-\Delta_{i,t}]\Delta^{-1}_{i,t}\rightarrow N(0,1)$ in distribution, we only need to show that the first term converges to zero.
	\begin{eqnarray}
		&&n^{1/2}\left[\widehat{\sigma}^{2}_{k,n,T}-n^{-1}\sigma^{2}_{k}\sum^{n-1}_{i=0}\{\Delta_{i}B_{k}(t)\}^{2}\Delta^{-1}_{i,t}\right]\nonumber\\
		&=&n^{-1/2}\left\{\sum^{n-1}_{i=0}\left(\Delta_{i}u_{k}(t)-\left[\hat{\dot{r}}_{k}+\sum^{N}_{l=1}\hat{a}_{kl}
		\exp\{u_{l}(t_{i})\}\right]\Delta_{i,t}\right)^{2}\Delta^{-1}_{i,t}
		-\sigma^{2}_{k}\sum^{n-1}_{i=0}\{\Delta_{i}B_{k}(t)\}^{2}\Delta^{-1}_{i,t}\right\}\nonumber\\
		&:=& n^{-1/2}(I'_{1}+I'_{2}),\nonumber
	\end{eqnarray}
	where
	\begin{eqnarray}
		I'_{1} &=&  \sum^{n-1}_{i=0}\Delta^{-1}_{i,t}\left(\int^{t_{i+1}}_{t_{i}}
		\sum^{N}_{l=1}\hat{a}_{kl}[\exp\{u_{l}(t)\}-\exp\{u_{l}(t_{i})\}]dt\right)^{2},\nonumber\\
		I'_{2} &=& -2\sigma_{k}\sum^{n-1}_{i=0}\Delta^{-1}_{i,t}\int^{t_{i+1}}_{t_{i}}dB_{k}(t)\cdot\int^{t_{i+1}}_{t_{i}}\sum^{N}_{l=1}\hat{a}_{kl}[\exp\{u_{l}(t)\}-\exp\{u_{l}(t_{i})\}]dt.\nonumber
	\end{eqnarray}
	Using the proof of Lemma \ref{lemma2}, It\^{o}'s isometry and the
	fact that $x_{k}(t)=e^{u_{k}(t)}$, it is easy to show
	\begin{eqnarray}
		E(I'_{1})&\leq& E\left\{\sum^{n-1}_{i=0}\int^{t_{i+1}}_{t_{i}}\left[\sum^{N}_{l=1}\hat{a}_{kl}\{x_{k}(t)-x_{k}(t_{i})\}\right]^{2}dt\right\}
		\leq  2 CT\Delta_{\max},\nonumber\\
		E(I'_{2})
		&\leq & 2 E \left\{\sum^{n-1}_{i=0}\left(\sigma^{2}_{k}\int^{t_{i+1}}_{t_{i}}dt\right)\right\}+
		\sum^{n-1}_{i=1}\Delta^{-2}_{i,t}E\left[\int^{t_{i+1}}_{t_{i}}\sum^{N}_{l=1}\hat{a}_{kl}\{x_{l}(t)-x_{l}(t_{i})\}dt\right]^{2}\nonumber\\
		&\leq& 2T\sigma^{2}_{k}+\sum^{N}_{l=1} \sum^{n-1}_{i=1}C\Delta^{-1}_{i,t}\left[\int^{t_{i+1}}_{t_{i}}E\{x_{l}(t)-x_{l}(t_{i})\}^{2}dt\right]\nonumber\\
		&\leq& 2T\sigma^{2}_{k}+ CNT,\nonumber
	\end{eqnarray}
	where the constant $C$ varies from line to line. So we have
	\begin{eqnarray}
		&&\sqrt{n}\left[\widehat{\sigma^{2}}_{k,
			n,T}-n^{-1}\sigma_{k}^2\sum^{n-1}_{i=0}\{\Delta_{i}B_{k}(t)\}^{2}\Delta^{-1}_{i,t}\right]
		\rightarrow_{p}0.\nonumber
	\end{eqnarray}
	This completes the proof of Theorem~\ref{FixedT}.
\end{proof}

Before the proof of Theorem~\ref{NoFixT},
we give the auxiliary lemma:

{\let\templemma\thelemma
	\renewcommand{\thelemma}{S2}
	
\begin{lemma}\label{lemma5}
Suppose Assumptions~\ref{assumptionA1}--\ref{assumptionA4} hold. For
fixed N, any $\tilde{\theta}=(\tilde{\vartheta},\sigma)\in {\mathcal
	K}\subseteq \Psi$, where $\mathcal{K}$ is a compact subset of
$\Psi=\Theta \times \Phi$,
any initial value $x(0)$, and $r=0,1,2$,
\begin{equation}
\sup_{\vartheta\in{\mathcal{K}}}|T^{-1}(\partial/\partial \vartheta)^{r}\ell_{n,T}(\vartheta)-T^{-1}(\partial/\partial \vartheta)^{r}\ell_{T}(\vartheta)|=O_{P_{\tilde{\theta}}}(\Delta^{1/2}_{\max}),T\rightarrow+\infty,n\rightarrow+\infty.
\end{equation}
\end{lemma}
}

\begin{proof}
	As in the proof of Lemma~\ref{lemma2}, we only consider the case when
	$r=0$, with other cases being similar. For any $M>0$,
	\begin{eqnarray}
		P_{\tilde{\theta}}\left\{T^{-1}\Delta^{-1/2}_{\max}\sup_{\vartheta \in\overline{\mathcal{K}}}|\ell_{n,T}(\vartheta)-\ell_{T}(\vartheta)|\geq M \right\}\leq M^{-1}E_{\tilde{\theta}}\left\{T^{-1}\Delta^{-1/2}_{\max}\sup_{\vartheta\in\overline{\mathcal{K}}}
		|\ell_{n,T}(\vartheta)-\ell_{T}(\vartheta)|\right\}.\nonumber
	\end{eqnarray}
	In order to prove the above inequality, we only need to show that for
	some constant $C$, we have $E_{\tilde{\theta}}\left\{T^{-1}\Delta^{-1/2}_{\max}\sup_{\vartheta\in\overline{\mathcal{K}}}|\ell_{n,T}(\vartheta)-\ell_{T}(\vartheta)|\right\}\leq C$.
	
	Since
	\begin{eqnarray}
		&& T^{-1}\Delta^{-1/2}_{\max}\{\ell_{n,T}(\vartheta)-\ell_{T}(\vartheta)\}\nonumber\\
		&=&T^{-1}\Delta^{-1/2}_{\max}\sum^{N}_{k=1}\sigma^{-2}_{k}\sum^{n-1}_{i=0}\int^{t_{i+1}}_{t_{i}}
		\sum^{N}_{l=1}a_{kl}\{x_{l}(t_{i})-x_{l}(t)\}
		\left\{\tilde{\dot{r}}_{k}+\sum^{N}_{l=1}\tilde{a}_{kl}x_{l}(t)\right\}dt\nonumber\\
		&&+T^{-1}\Delta^{-1/2}_{\max}\sum^{N}_{k=1}\sigma^{-1}_{k}\sum^{n-1}_{i=0}
		\int^{t_{i+1}}_{t_{i}}\sum^{N}_{l=1}a_{kl}\{x_{l}(t_{i})
		-x_{l}(t)\}dB_{k}(t)\nonumber\\
		&&-(2T)^{-1}\Delta^{-1/2}_{\max}\sum^{N}_{k=1}\sigma^{-2}_{k}\sum^{n-1}_{i=0}
		\int^{t_{i+1}}_{t_{i}}\left\{\dot{r}_{k}+\sum^{N}_{l}a_{kl}x_{l}(t_{i})\right\}^{2}-\left\{\dot{r}_{k}+\sum^{N}_{l}a_{kl}x_{l}(t)\right\}^{2}dt\nonumber\\
		&:= & \mathrm{S}_{1}+\mathrm{S}_{2}+\mathrm{S}_{3},\nonumber
	\end{eqnarray}
	where $\mathrm{S}_{1},\mathrm{S}_{2},\mathrm{S}_{3}$ correspond to the three terms of the above
	equation respectively. It is easy to show that
	\begin{eqnarray}
		E\sup_{\vartheta\in {\mathcal{K}}}|\mathrm{S}_{1}|
		&\leq& E\sup_{\vartheta \in {\mathcal{K}}}\left|T^{-1}\Delta^{-1/2}_{\max}
		\sum^{N}_{k=1}\sigma^{-2}_{k}\sum^{n-1}_{i=0}\int^{t_{i+1}}_{t_{i}}
		\sum^{N}_{l=1}a_{kl}\{x_{l}(t_{i})-x_{l}(t)\}
		\left\{\tilde{R}_{k}+\sum^{N}_{l=1}\tilde{a}_{kl}x_{l}(t)\right\}dt\right|, \nonumber\\
		&\leq& T^{-1}\Delta^{-1/2}_{\max}
		\sum^{N}_{k=1}\sum^{N}_{l=1}\sigma^{-2}_{k}\sum^{n-1}_{i=0}\left(\sup_{\vartheta\in {\mathcal{K}}}|a_{kl}\tilde{R}_{k}|\right)\int^{t_{i+1}}_{t_{i}}
		E|\{x_{l}(t_{i})-x_{l}(t)\}|dt, \nonumber\\
		&&+T^{-1}\Delta^{-1/2}_{\max}
		\sum^{N}_{k=1}\sum^{N}_{l,m=1}\sigma^{-2}_{k}\sum^{n-1}_{i=0}\left(\sup_{\vartheta\in {\mathcal{K}}}|a_{kl}a_{km}\tilde{R}_{k}|\right)\int^{t_{i+1}}_{t_{i}}
		E|\{x_{l}(t_{i})-x_{l}(t)\}x_{l}|dt, \nonumber
	\end{eqnarray}
	\begin{eqnarray}
		E\sup_{\vartheta\in {\mathcal{K}}}|\mathrm{S}_{2}|&\leq &E\sup_{\vartheta \in{\mathcal{K}}}\left[T^{-1}\Delta^{-1/2}_{\max}\sum^{N}_{k=1}\sigma^{-1}_{k}\sum^{n-1}_{i=0}
		\bigg|\int^{t_{i+1}}_{t_{i}}\sum^{N}_{l=1}a_{kl}\{x_{l}(t_{i})-x_{l}(t)\}dB_{k}(t)\bigg|\right] \nonumber\\
		&\leq &T^{-1}\Delta^{-1/2}_{\max}\sum^{N}_{k=1}\sigma^{-1}_{k}\sum^{n-1}_{i=0}\sum^{N}_{l=1}
		\left(\sup_{\vartheta\in {\mathcal{K}}}|a_{kl}|\right)E\left|\int^{t_{i+1}}_{t_{i}}\{x_{l}(t_{i})-x_{l}(t)\}dB_{k}(t)\right| \nonumber\\
		&\leq&
		T^{-1}\Delta^{-1/2}_{\max}\sum^{N}_{k=1}\sigma^{-1}_{k}\sum^{n-1}_{i=0}\sum^{N}_{l=1}
		\left(\sup_{\vartheta\in {\mathcal{K}}}|a_{kl}|\right)\left\{E\left(\int^{t_{i+1}}_{t_{i}}\{x_{l}(t_{i})-x_{l}(t)\}^{2}dt\right) \right\}^{1/2}, \nonumber
	\end{eqnarray}
	\begin{eqnarray}
		E\sup_{\vartheta\in{\mathcal{K}}}|\mathrm{S}_{3}|&\leq & E\sup_{\vartheta\in {\mathcal{K}}}\left[(2T)^{-1}\Delta^{-1/2}_{\max}\sum^{N}_{k=1}\sigma^{-2}_{k}\sum^{n-1}_{i=0}\int^{t_{i+1}}_{t_{i}}
		\left|R_{k}+\sum^{N}_{l=1}a_{kl}x_{l}(t_{i})\right|\left|\sum^{N}_{l=1}a_{kl}\{x_{l}(t_{i})-x_{l}(t)\}\right|dt\right] \nonumber\\
		&& +E\sup_{\vartheta\in {\mathcal{K}}}\left[2T^{-1}\Delta^{-1/2}_{\max}\sum^{N}_{k=1}\sigma^{-2}_{k}\sum^{n-1}_{i=0}\int^{t_{i+1}}_{t_{i}}
		\left|R_{k}+\sum^{N}_{l=1}a_{kl}x_{l}(t)\right|\left|\sum^{N}_{l=1}a_{kl}\{x_{l}(t_{i})-x_{l}(t)\}\right|dt\right].\nonumber\\
		&\leq&T^{-1}\Delta^{-1/2}_{\max}\sum^{N}_{k=1}\sigma^{-2}_{k}\sum^{n-1}_{i=0}
		\sum^{N}_{l=1} \left(\sup_{\vartheta\in {\mathcal{K}}}|R_{k}a_{kl}|\right)
		\int^{t_{i+1}}_{t_{i}}E|x_{l}(t_{i})-x_{l}(t)|dt \nonumber\\
		&&+(2T)^{-1}\Delta^{-1/2}_{\max}\sum^{N}_{k=1}\sigma^{-2}_{k}
		\sum^{N}_{l,m=1}\left(\sup_{\vartheta\in{\mathcal{K}}}|a_{kl}a_{km}|\right)\sum^{n-1}_{i=0}\int^{t_{i+1}}_{t_{i}}
		E|x_{l}(t_{i})\{x_{l}(t_{i})-x_{l}(t)\}|dt\nonumber\\
		&&+(2T)^{-1}\Delta^{-1/2}_{\max}\sum^{N}_{k=1}\sigma^{-2}_{k}
		\sum^{N}_{l,m=1}\left(\sup_{\vartheta\in{\mathcal{K}}}|a_{kl}a_{km}|\right)\sum^{n-1}_{i=0}\int^{t_{i+1}}_{t_{i}}
		E|x_{l}(t)\{x_{l}(t_{i})-x_{l}(t)\}|dt.\nonumber
	\end{eqnarray}
	Since $\mathcal{K}$ is a compact set, the terms, $\sup_{\vartheta\in{\mathcal{K}}}|\tilde{R}_{k}a_{kl}|$,
	$\sup_{\vartheta\in{\mathcal{K}}}|\tilde{R}_{k}a_{kl}a_{km}|$,
	$\sup_{\vartheta\in{\mathcal{K}}}|a_{kl}|$,
	$\sup_{\vartheta\in{\mathcal{K}}}|R_{k}a_{kl}|$,
	$\sup_{\vartheta\in{\mathcal{K}}}|a_{kl}a_{km}|$,
	are all bounded.
	In order to prove the boundedness of $E\sup_{\vartheta\in
		{\mathcal{K}}}|\mathrm{S}_{1}|$, $E\sup_{\vartheta\in
		{\mathcal{K}}}|\mathrm{S}_{2}|$ and $E\sup_{\vartheta\in {\mathcal{K}}}|\mathrm{S}_{3}|$, we only need to show that there exist two constants $C_{1}, C_{2}$, such that
	\begin{equation}\label{3.3.3.1}
		E|\{x_{l}(t_{i})-x_{l}(t)\}x_{m}(t)| \leq  C_{1}\Delta_{\max}+C_{2}\Delta^{1/2}_{\max}.
	\end{equation}
	
	By Proposition~\ref{pro2.2},
	\begin{eqnarray}
		&&E|\{x_{l}(t)-x_{l}(t_{i})\}x_{m}(t)|\nonumber\\
		&=& E\bigg|\int^{t}_{t_{i}}\{r_{k}+\sum^{N}_{l=1}a_{kl}x_{l}(s)\}x_{k}(s)x_{m}(t)ds
		+\int^{t}_{t_{i}}\sigma_{k}x_{k}(s)x_{m}(t)dB_{k}(s)\bigg|\nonumber\\
		&\leq&
		E\bigg|
		\int^{t}_{t_{i}}\{r_{k}+\sum^{N}_{l=1}a_{kl}x_{l}(s)\}x_{k}(s)x_{m}(t)ds\bigg|
		+\sigma_{k}\left[\int^{t}_{t_{i}}E\{x^{2}_{m}(t)x^{2}_{k}(s)\}ds\right]\nonumber\\
		&\leq & C_{1}\Delta_{\max}+C_{2}\Delta^{1/2}_{\max},\nonumber
	\end{eqnarray}
	(\ref{3.3.3.1}) is proved.
	
	Because $E\sup_{\vartheta \in {\mathcal{K}}}|\mathrm{S}_{1}|, E\sup_{\vartheta \in {\mathcal{K}}}|\mathrm{S}_{2}|$ and $E\sup_{\vartheta \in {\mathcal{K}}}|\mathrm{S}_{3}|$ are bounded, we can conclude that there exist constants
	$C>0,T_{0}\geq0$, and $n_{0}\in$
	such that for any $T>T_{0}$, $n\geq n_{0}$, $M>0$,
	\begin{eqnarray}
		&&P_{\tilde{\theta}}\left\{\Delta^{-1/2}_{\max}\sup_{\vartheta\in\in{\mathcal{K}}}\left|T^{-1}\ell_{n,T}(\vartheta)-T^{-1}\ell_{T}(\vartheta)\right|\geq M \right\}\leq CM^{-1},\nonumber
	\end{eqnarray}
	and
	\begin{eqnarray}
		\lim_{M\rightarrow
			+\infty}\limsup_{n\rightarrow+\infty,T\rightarrow+\infty}P_{\tilde{\theta}}\left\{\Delta^{-1/2}_{\max}\sup_{\vartheta\in{\mathcal{K}}}\left|T^{-1}\ell_{n,T}(\vartheta)-T^{-1}\ell_{T}(\vartheta)\right|\geq M\right\} &=& 0.\nonumber
	\end{eqnarray}
\end{proof}

\begin{theorem}\label{SupplementaryNoFixT}
	Under Assumptions~\ref{assumptionA1}--\ref{assumptionA4} and~\ref{assumptionA6}(I), we have
	\begin{align*}
		\|\hat{\vartheta}_{n,T}-\hat{\vartheta}_{T}\|=O_{p}(\Delta^{1/2}_{\max}),\quad
		\|\hat{\vartheta}_{n,T}-\vartheta^{0}\|=o_{p}(1).
	\end{align*}
	If we further assume Assumption~\ref{assumptionA6}(II), then
	\begin{align*}
		T^{1/2}(\hat{\vartheta}_{n,T}-\vartheta^{0}) \rightarrow N(0, 
		(n/2)^{1/2}\sigma^{-2}_{k}(\widehat{\sigma}^{2}_{k,n,T}-\sigma^{2}_{k}) \rightarrow N(0,1)
	\end{align*}
	where $\mathcal{I}(\vartheta)={\rm diag}\{\sigma^{-2}_{1} I_{1}(\vartheta_{1}),...,\sigma^{-2}_{N}I_{N}(\vartheta_{N})\}$, $I_{k}(\vartheta_{k})$ is defined in (\ref{Information}), $k=1,...,N$.
\end{theorem}

\begin{proof}
	From the proof of Theorem~\ref{FixedT},
	we have
	\begin{eqnarray}
		&&\|\hat{\vartheta}_{n,T}-\hat{\vartheta}_{T}\|\nonumber\\ &\leq &  \|\{\ell''_{T}(\hat{\vartheta}_{T})\}^{-1}\|\cdot\|\ell''_{T}(\hat{\vartheta}_{T})(\hat{\vartheta}_{n,T}-\hat{\vartheta}_{T})\|\nonumber\\
		&=&
		\|\{\ell''_{T}(\hat{\vartheta}_{T})\}^{-1}\|\cdot\|\ell'_{T}(\hat{\vartheta}_{n,T})-\ell'_{T}(\hat{\vartheta}_{T})+\ell'_{n,T}(\hat{\vartheta}_{n,T})-\ell'_{T}(\hat{\vartheta}_{n,T})-\ell''_{T}(\hat{\vartheta}_{T})(\hat{\vartheta}_{n,T}-\hat{\vartheta}_{T})\|\nonumber\\
		&\leq &
		\|\{\ell''_{T}(\hat{\vartheta}_{T})\}^{-1}\|\cdot\|\ell'_{T}(\hat{\vartheta}_{n,T})-\ell'_{T}(\hat{\vartheta}_{T})-\ell''_{T}(\hat{\vartheta}_{T})(\hat{\vartheta}_{n,T}-\hat{\vartheta}_{T})\|\nonumber\\
		&&+ \|\{\ell''_{T}(\hat{\vartheta}_{T})\}^{-1}\|\cdot\|\ell'_{n,T}(\hat{\vartheta}_{n,T})-\ell'_{T}(\hat{\vartheta}_{n,T})\|.\nonumber
	\end{eqnarray}
	By inspection, we see that $\ell_{T}(\vartheta)$ is a quadratic
	function in $\vartheta$, so
	$$\ell'_{T}(\hat{\vartheta}_{n,T})-\ell'_{T}(\hat{\vartheta}_{T})-\ell''_{T}(\hat{\vartheta}_{T})(\hat{\vartheta}_{n,T}-\hat{\vartheta}_{T})=0.$$
	By Lemma~\ref{lemma5}, for any compact subset $K\subseteq \Theta$, and
	any $\epsilon>0$,
	there exists a constant $C(\epsilon)$ such that
	\begin{equation}\label{applylemma5}
		P\left(\sup_{\vartheta \in
			{\mathcal{K}}}\|T^{-1}\ell'_{n,T}(\vartheta)-T^{-1}\ell'_{T}(\vartheta)\|
		>  C(\epsilon)\Delta^{-1/2}_{\max}\right)<\epsilon.
	\end{equation}
	
	Let $\lambda_{0}>0$ be the minimal eigenvalue of the Fisher
	information matrix $I(\vartheta)$. By Theorem 4.5(i) from
	\cite{MH:2017} or Theorem 2 of \cite{BH:1975}, there a.s. exists $T_{0}>0$,
	such that for all $T>T_{0}$ and all $\hat{\vartheta}_{T}$ in the neighborhood
	of $\vartheta$ with the radius $\epsilon_{0}/2$,
	$\min_{\|\vartheta_{\epsilon}\|=1}\vartheta_{\epsilon}^{\mathcal{T}}\{-T^{-1}\ell''_{T}(\hat{\vartheta}_{T})\}\vartheta_{\epsilon}\geq
	\lambda_{0}/2>0$.  So we have
	$\|\{\ell''_{T}(\hat{\vartheta}_{T})\}^{-1}\| \leq
	2\lambda^{-1}_{0}T^{-1}$. Combining this with \eqref{applylemma5}, for
	sufficiently large $T$,
	$$ P\left(\|\hat{\vartheta}_{n,T}-\hat{\vartheta}_{T}\| >  2\lambda^{-1}_{0}T^{-1}C(\epsilon)T\Delta^{1/2}_{\max}\right)<\epsilon.$$
	For any $\delta>0$, we can choose a sequence $\epsilon_n\rightarrow 0$
	such that
	$C_n:=2C(\epsilon_n)\lambda^{-1}_{0}<C'\Delta_{\max}^{-\delta}$. This gives
	\begin{eqnarray}
		\limsup_{n,T}P_{\theta}\left(\left[\left\|\hat{\vartheta}_{n,T}-\hat{\vartheta}_{T}\right\|\leq C' \Delta^{\frac{1}{2}-\delta}_{\max}\right]^{c} \right)= 0.\nonumber
	\end{eqnarray}
	Hence we get $\hat{\vartheta}_{n,T}\rightarrow_{p}\hat{\vartheta}_{T}$. From the above results and Slutsky's
	theorem, when $T\rightarrow+\infty,\Delta_{\max}\rightarrow0$, we have the consistency of $\hat{\vartheta}_{n,T}$,
	\begin{eqnarray}
		\|\hat{\vartheta}_{n,T}-\vartheta\| \leq   \|\hat{\vartheta}_{n,T}-\hat{\vartheta}_{T}\|+\|\hat{\vartheta}_{T}-\vartheta\|\nonumber
		\rightarrow  0.
	\end{eqnarray}
	To prove its asymptotic normality, we consider
	\begin{eqnarray}
		\|T^{1/2}(\hat{\vartheta}_{n,T}-\vartheta)-T^{1/2}(\hat{\vartheta}_{T}-\vartheta)\| = (T\Delta_{\max})^{1/2}\Delta^{-1/2}_{\max}\|\hat{\vartheta}_{n,T}-\hat{\vartheta}_{T}\| \leq  C(\epsilon_n)T^{1/2}\Delta^{1/2}_{\max}
		\nonumber
	\end{eqnarray}
	Now since $\lim_{n,T}T\Delta_{\max}=0$, we can choose a sequence
	$\epsilon_n$ such that $\lim_{n,T}T\Delta_{\max}C(\epsilon_n)^2=0$. Then by Slutsky's theorem and the asymptotic normality of $T^{1/2}(\hat{\vartheta}_{T}-\vartheta)$, we get the asymptotic normality of $\hat{\vartheta}_{n,T}$.
	
	Next, we consider the asymptotic property of $\widehat{\sigma}^{2}_{k,n,T}$.
	\begin{eqnarray}
		&& \sqrt{n}(\widehat{\sigma}^{2}_{k,n,T}-\sigma^{2}_{k})\nonumber\\
		&=& \sqrt{n}\left[\widehat{\sigma}^{2}_{k,n,T}-n^{-1}\sigma^{2}_{k}\sum^{n-1}_{i=0}\{\Delta_{i}B_{k}(t)\}^{2}\Delta^{-1}_{i,t}\right]
		+\sqrt{2}\sigma^{2}_{k}(2n)^{-1/2}\sum^{n-1}_{i=0}[\{\Delta_{i}B_{k}(t)\}^{2}-\Delta_{i,t}]\Delta^{-1}_{i,t}.\nonumber
	\end{eqnarray}
	Since
	$(2n)^{-1/2}\sum^{n-1}_{i=0}[\{\Delta_{i}B_{k}(t)\}^{2}-\Delta_{i,t}]\Delta^{-1}_{i,t}\rightarrow_{L}
	N(0,1)$, we only need to show that the first term converges to
	zero. Let $\bar{\Delta}=n^{-1}T$: note that $\Delta_{\max}\geq\bar{\Delta}$.
	\begin{eqnarray}
		&&n^{1/2}\left[\widehat{\sigma}^{2}_{k,n,T}-n^{-1}\sigma^{2}_{k}\sum^{n-1}_{i=0}\{\Delta_{i}B_{k}(t)\}^{2}\Delta^{-1}_{i,t}\right]\nonumber\\
		&=&(T\bar{\Delta})^{1/2}T^{-1}\left\{\sum^{n-1}_{i=0}\left(\Delta_{i}u_{k}(t)-\left[\hat{\dot{r}}_{k}
		+\sum^{N}_{l=1}\hat{a}_{kl}\exp\{u_{l}(t_{i})\}\right]\Delta_{i,t}\right)^{2}\Delta^{-1}_{i,t}
		-\sigma^{2}_{k}\sum^{n-1}_{i=0}\{\Delta_{i}B_{k}(t)\}^{2}\Delta^{-1}_{i,t}\right\}\nonumber\\
		&:=& (T\bar{\Delta})^{1/2}(S'_{1}+S'_{2}),\nonumber
	\end{eqnarray}
	where
	\begin{eqnarray}
		S'_{1} &=& T^{-1} \sum^{n-1}_{i=0}\Delta^{-1}_{i,t}\left(\int^{t_{i+1}}_{t_{i}}
		\sum^{N}_{l=1}\hat{a}_{kl}[\exp\{u_{l}(t)\}-\exp\{u_{l}(t_{i})\}]dt\right)^{2},\nonumber\\
		S'_{2} &=& -2\sigma_{k}T^{-1}\sum^{n-1}_{i=0}\Delta^{-1}_{i,t}\int^{t_{i+1}}_{t_{i}}dB_{k}(t)\cdot\int^{t_{i+1}}_{t_{i}}\sum^{N}_{l=1}\hat{a}_{kl}
		[\exp\{u_{l}(t)\}-\exp\{u_{l}(t_{i})\}]dt,\nonumber
	\end{eqnarray}
	From the proof of Corollary~\ref{pro2.4}
	and It\^{o}'s isometry, we have
	\begin{eqnarray}
		ES'_{1}&\leq& 2T^{-1}E\left\{\sum^{n-1}_{i=0}\int^{t_{i+1}}_{t_{i}}\left[\sum^{N}_{l=1}\hat{a}_{kl}\{x_{k}(t)-x_{k}(t_{i})\}\right]^{2}dt\right\}
		\leq  2 C\Delta_{\max},\nonumber\\
		ES'_{2}
		&\leq & 2T^{-1} E \left\{\sum^{n-1}_{i=0}\left(\sigma^{2}_{k}\int^{t_{i+1}}_{t_{i}}dt\right)\right\}+
		T^{-1}\sum^{n-1}_{i=1}\Delta^{-2}_{i,t}E\left(\int^{t_{i+1}}_{t_{i}}\sum^{N}_{l=1}\hat{a}_{kl}
		[\exp\{u_{l}(t)\}-\exp\{u_{l}(t_{i})\}]dt\right)^{2},\nonumber\\
		&\leq & 2\sigma^{2}_{k}+ CNT^{-1}\sum^{n-1}_{i=1}\Delta^{-1}_{i,t}\int^{t_{i+1}}_{t_{i}}E[\exp\{u_{l}(t)\}-\exp\{u_{l}(t_{i})\}]^{2}dt,\nonumber\\
		&\leq & 2\sigma^{2}_{k}+ CN.\nonumber
	\end{eqnarray}
	So we have
	\begin{eqnarray}
		&&n^{1/2}\left[\widehat{\sigma}^{2}_{k,n,T}-n^{-1}\sigma^{2}_{k}\sum^{n-1}_{i=0}\{\Delta_{i}B_{k}(t)\}^{2}\Delta^{-1}_{i,t}\right]
		\rightarrow0.\nonumber
	\end{eqnarray}
	in probability. This completes the proof of Theorem 2.
\end{proof}

\end{document}